\documentclass[10pt]{article}

\usepackage{amsmath,amssymb,amsfonts}
\usepackage{physics}

\usepackage{geometry}

\usepackage{amsthm,bm,mathtools}
\usepackage{authblk}
\usepackage{thm-restate}
\usepackage{algorithm}
\usepackage{algpseudocode}
\usepackage{adjustbox}
\usepackage{tablefootnote}

\newtheorem{theorem}{Theorem}
\newtheorem{lemma}[theorem]{Lemma}
\newtheorem{definition}[theorem]{Definition}

\newtheorem{assumption}[theorem]{Assumption}

\newtheorem{remark}[theorem]{Remark}

\usepackage[utf8]{inputenc}

\usepackage{graphicx}

\usepackage{bbm}
\usepackage{fullpage}
\usepackage{soul}
\usepackage[colorlinks,linkcolor=blue,citecolor=blue,urlcolor=magenta,filecolor=cyan]{hyperref}
\usepackage[capitalize]{cleveref}
\usepackage{xcolor}

\crefname{assumption}{Assumption}{Assumptions}

\newcommand{\x}{\boldsymbol{x}}

\newcommand{\z}{\boldsymbol{z}}

\newcommand{\bbr}{\mathbb{R}}

\title{Stochastic Quantum Sampling for Non-Logconcave Distributions and Estimating Partition Functions}

\author[$\dag$]{Guneykan Ozgul}
\author[*]{Xiantao Li}
\author[$\dag$]{Mehrdad Mahdavi}
\author[$\dag$]{Chunhao Wang}

\affil[$\dag$]{Department of Computer Science and Engineering, Pennsylvania State University}
\affil[*]{Department of Mathematics, Pennsylvania State University}
\affil[ ]{\{\texttt{{gmo5119,xiantao.li,mzm616,cwang\}@psu.edu}}}

\date{}

\begin{document}

\maketitle

\begin{abstract}
We present quantum algorithms for sampling from non-logconcave probability distributions in the form of $\pi(x) \propto \exp(-\beta f(x))$. Here, $f$ can be written as a finite sum $f(x)\coloneqq \frac{1}{N}\sum_{k=1}^N f_k(x)$. Our approach is based on quantum simulated annealing on slowly varying Markov chains derived from unadjusted Langevin algorithms, removing the necessity for function evaluations which can be computationally expensive for large data sets in mixture modeling and multi-stable systems. We also incorporate a stochastic gradient oracle that implements the quantum walk operators inexactly by only using mini-batch gradients. As a result, our stochastic gradient based algorithm only accesses small subsets of data points in implementing the quantum walk. One challenge of quantizing the resulting Markov chains is that they do not satisfy the detailed balance condition in general. Consequently, the mixing time of the algorithm cannot be expressed in terms of the spectral gap of the transition density, making the quantum algorithms nontrivial to analyze. To overcome these challenges, we first build a hypothetical Markov chain that is reversible, and also converges to the target distribution. Then, we quantified the distance between our algorithm's output and the target distribution by using this hypothetical chain as a bridge to establish the total complexity. Our quantum algorithms exhibit polynomial speedups in terms of both dimension and precision dependencies when compared to the best-known classical algorithms.
\end{abstract}

\section{Introduction}
Many problems in statistics, physics, finance, machine learning,  optimization, and molecular dynamics involve sampling from a distribution with a density proportional to $e^{-\beta f(x)}$, known as a Gibbs distribution. For instance, techniques for sampling from such a distribution play a central role in statistical mechanics in probing  equilibrium states of physical systems, understanding phase transition, and estimating thermodynamic properties \cite{chandler1987introduction}. In machine learning, sampling from these distributions aids in exploring the posterior distribution in Bayesian machine learning~\cite{neal2011mcmc,ahn2012bayesian,cheng2018underdamped}, enabling parameter estimation, uncertainty quantification, and model comparison~\cite{murphy2022probabilistic}.  In convex geometry, an effective sampling strategy is central to estimating the volume of a convex body that can be  applied to problems in statistics, theoretical computer science, and operations research~\cite{vempala2007geometric}.

A well-known classical method for Gibb's sampling is the Markov Chain Monte Carlo (MCMC) method, where a Markov chain with the desired stationary density is constructed. Therefore samples can be generated by running the Markov chain for sufficiently a long time. One such Markov chain can be obtained through careful discretization of Langevin diffusion equation and it inspired a large family of gradient-based sampling algorithms. 

Langevin diffusion equation is a stochastic differential equation, as will be described in the next section, and it converges to the desired Gibbs distribution. A simple yet common discretization of Langevin diffusion is to use Euler-Maruyama method with sufficiently small step size $\eta>0$ and this algorithm is called the over-damped Langevin algorithm. However, due to finite-sized discretization, the Markov chain is asymptotically biased. That is, it only converges to the neighborhood of the desired Gibbs distribution, prohibiting one from using large step sizes because of the discrepancy. To overcome this bias, one can adjust the Markov chain by introducing the Metropolis-Hastings filter, which is used as a conditional rejection method to guarantee that the chain is time-reversible and it converges to the desired distribution. Such algorithms are sometimes referred to as Metropolis-adjusted Langevin algorithms (MALA), and the algorithms without the rejection step are conventionally called unadjusted Langevin algorithms (ULA). 

In the past decade, notable progress has been witnessed on the theoretical development of quantum algorithms for various machine learning and optimization problems. It is natural to hope that quantum computers also provide provable speedups for general sampling problems. If we could prepare a quantum state whose amplitudes correspond to some desired distribution, then measuring this state yields a random sample from this probability distribution. Unfortunately, quantum speedups in such sampling models probably does not hold in general as this will imply $\mathsf{SZK} \subseteq \mathsf{BQP}$~\cite{AT03}. Whereas the hardness barrier exists for a quantum speedup for general sampling problems, in some special cases, it has been shown that  quantum algorithms can achieve polynomial speedups over classical algorithms. Such examples include quantum algorithms for uniformly sampling on a 2D lattice~\cite{Richter2007}, 
for estimating partition functions~\cite{wocjan2008speedup,wocjan2009quantum,montanaro2015quantum,harrow2019adaptive,AHN21}, and for estimating volumes of convex bodies~\cite{CCH+19}.

Recently, Childs, Li, Liu, Wang, and Zhang~\cite{childs2022quantum} introduced a quantum MALA algorithm based on quantum simulated annealing, that leverages the fact that a coherent quantum state corresponding to desired logconcave distribution can be prepared using fewer number of calls to gradient and evaluation oracle than the classical counterparts. Can we achieve quantum speedups for more general distributions, e.g., non-logconcave distributions? Moreover, an open question in~\cite{childs2022quantum} was how to speed up ULA using similar techniques. In this paper, we address these questions by giving quantum algorithms for sampling from non-logconcave distributions based on ULA.

To deal with non-logconcave distribution, one challenge is to analyze the cooling schedule, which is not needed in classical sampling algorithms. To tackle this challenge, we use isoperimetric inequalities for non-lonconcave distributions instead of concentration inequalities used in the analysis for logconcave distributions. For quantizing ULA, another significant challenge is that the underlying Markov chain does not satisfy detailed balance condition due to lack of the Metropolis-Hastings filter. On the other hand, current quantum algorithmic techniques rely on this assumption and quantizing a non-reversible Markov chain is an open problem. To overcome this challenge, we use a quantum MALA algorithm as a theoretical bridge and quantify our algorithm's error in terms of step size and other parameters. Since quantum MALA is time-reversible and asymptotically unbiased, it converges to target distribution, allowing us to express our algorithm's error with respect to Gibbs distribution. In the construction of our algorithms, we use standard quantum simulated annealing techniques as in~\cite{childs2022quantum} while the underlying Markov chain is non-reversible. Then we perform a perturbation analysis to bound the error.  Although this analysis technique is used in classical analysis of Markov chains, to the best of our knowledge, this is the first time it is applied to analyze quantum walks for non-reversible Markov chains.

\paragraph{Problem formulation} In this paper, we focus on designing and analyzing quantum algorithms for sampling the Gibbs density $\propto e^{-\beta f(x)}$ where $f(x)$ is not necessarily convex. 
An important scenario in machine learning is when $f(x)$ admits a decomposition, 
\begin{equation}
    f(x) =\frac1N \sum_{k=1}^N f_k(x), 
\end{equation}
where $N \gg 1$ is large. One typical example is where $x$ comes from the model parameter, and $ f(x)$ is the empirical loss defined on a large data set. Clearly, this will cause a significant slowdown of existing quantum algorithms.

We further make the following assumptions on the nonconvex functions $f$. These assumptions are realistic as they are satisfied in many applications and are widely assumed in the literature.
\begin{assumption}[Smoothness]
    \label{assumption1}
    There exists a positive constant $L$ such that for any $x,y \in \mathbb{R}^d$ and all functions $f_k(x),k=1,...,n$, it holds that
    \begin{equation}
        \|\nabla f_k(x) - \nabla f_k(y)\|\leq L\|x-y\|.
    \end{equation}
\end{assumption}
\begin{assumption}[Dissipativeness]
    \label{assumption2}
    There are absolute constants $m>0$ and $b\geq 0$ such that
    \begin{equation}
        \langle \nabla f(x),x \rangle\geq m\|x\|^2-b.
    \end{equation}
\end{assumption}
These are standard assumptions used in sampling and optimization literature in non-convex settings. The first assumption ensures that small changes in the input parameters result in bounded changes in the function values whereas the second one implies that the gradients of the objective function are well-behaved and do not explode, helping in the stability of the algorithm. 
Next, we give the following definitions that are commonly used in the analysis of non-logconcave sampling.
\begin{definition}[Cheeger Constant]
    Let $\nu$ be a probability measure on $\Omega$. Then $\nu$ satisfies the isoperimetric inequality with Cheeger constant $\rho$ if for any $A\in \Omega$, it holds that 
    \begin{align}
        \lim\limits_{h\to 0^{+}}\inf \frac{\nu(A_h) -\nu(A)}{h}\geq \rho \min\{\nu(A),1-\nu(A) \},
    \end{align}
   where $A_h=\{x\in \Omega:\exists y\in A, \|x-y\|\leq h\}$ .
\end{definition}

\begin{definition}[Log-Sobolev Inequality]
       Let $\nu$ be a probability measure on $\Omega$. We say that $\nu$ satisfies Log-Sobolev inequality with constant $c_{\mathrm{LSI}}$ if for any smooth function $g$ on $\mathbb{R}^d$, satisfying $\int_x g(x)\nu(x)\,\dd x=1$, it holds that
       \begin{align}
           \int g(x)\log(g(x))\nu(x)\,\dd x\leq \frac{1}{2c_{\mathrm{LSI}}}\int\frac{\|\nabla g(x)\|^2}{g(x)}\nu(x)\,\dd x.
       \end{align}
\end{definition}

\paragraph{Oracle model} We assume that we have the access to the following oracles to implement our algorithm. These oracles are virtually classical oracles while empowering superposition access. We first define the full gradient oracle for $f$ as follows:
\begin{equation}
  \mathcal{O}_{\nabla f}\ket{x}\ket{0} = \ket{x}\ket{\nabla f(x)}. 
\end{equation}
Similarly, we define a stochastic gradient oracle, 
\begin{equation}
  \mathcal{O}_{\Tilde{\nabla} f}\ket{x}\ket{0} = \ket{x}\ket{\Tilde{\nabla} f(x)}. 
\end{equation}
where  $\tilde{\nabla} f(x) = \frac{1}{B}\sum\limits_{k\in S} \nabla f_k(x)$ where $S$ is a subset of size $B$ data samples chosen randomly without replacement. Note that $O_{\Tilde{\nabla} f}$ possibly outputs a different state for the same input state depending on the internal random batch. Finally, the evaluation oracle is defined:
\begin{equation}
  \mathcal{O}_{f}\ket{x}\ket{0} = \ket{x}\ket{f(x)}. 
\end{equation}
We note that although we quantified the complexity of our algorithm in terms of the number of calls to these oracles, the evaluation oracle and full gradient oracle are slower than the stochastic gradient oracle due to the finite sum of $N$ terms. Hence, our contribution is to use a stochastic gradient oracle to improve quantum walk implementation.

\paragraph{Main contributions}
To solve the sampling problem, we present three quantum algorithms. The first algorithm is based on MALA, which requires the full gradient oracle and the function evaluation oracle. The result is concluded in the following oracle, which is proved in \cref{sec:q-mala}
\begin{restatable}[Quantum MALA]{theorem}{qmala}
  \label{thm:quantum-mala}
Let $\pi \propto e^{-\beta f(x)} $ denote a probability distribution with inverse temperature $\beta>0$ such that $f(x)$ satisfies \cref{assumption1,assumption2}. Then, there exists a quantum algorithm that outputs a random variable distributed according to $\mu$ such that,
\begin{align}
\|\mu - \pi \|_{\mathrm{\mathrm{TV}}}\leq \epsilon,
\end{align}
where $\|.\|_{\mathrm{\mathrm{TV}}}$ is the total variation distance, using $\Tilde{O}\left(\beta d \rho^{-1}c_{\mathrm{LSI}}^{-1} \right)$ queries\footnote{Throughout this paper, we use the notation $\tilde{O}(\cdot)$ to hide the poly-logarithmic dependencies on $\beta,\epsilon, d, \rho$, and $c_{\mathrm{LSI}}$.} to $\mathcal{O}_{\nabla f}$ and $\mathcal{O}_f$.
\end{restatable}

The second sampling algorithm is based on ULA, which requires the full gradient oracle. The result is stated in the following theorem, which is proved in \cref{sec:q-ula}.
\begin{restatable}[Quantum ULA]{theorem}{qula}
\label{thm:quantum-ula}
Let $\pi \propto e^{-\beta f(x)} $ denote a probability distribution with inverse temperature $\beta>0$ such that $f(x)$ satisfies \cref{assumption1,assumption2}. Then, there exists a quantum algorithm that outputs a random variable distributed according to $\mu$ such that,
\begin{align}
\|\mu - \pi \|_{\mathrm{\mathrm{TV}}}\leq \epsilon,
\end{align}
where $\|.\|_{\mathrm{\mathrm{TV}}}$ is the total variation distance, using $\Tilde{O}\left(\beta d^{3/2}\epsilon^{-1} \rho^{-1} c_{\mathrm{LSI}}^{-1}\right)$ queries to $\mathcal{O}_{\nabla f}$.
\end{restatable}

The last sampling algorithm is also based on ULA while we use the stochastic oracle instead of the full oracle, which makes the implementation easier. We state the result in the following theorem, which is proved in \cref{sec:q-sula}.
\begin{restatable}[Quantum ULA with stochastic gradient]{theorem}{qsula}
\label{thm:quantum-sula}
Let $\pi \propto e^{-\beta f(x)} $ denote a probability distribution with inverse temperature $\beta>0$ such that $f(x) = \frac1N \sum\limits_{k=1}^N f_k(x)$ satisfies \cref{assumption1,assumption2}. Then, there exists a quantum algorithm that outputs a random variable distributed according to $\mu$ such that,
\begin{align}
\|\mu - \pi \|_{\mathrm{\mathrm{TV}}}\leq \epsilon,
\end{align}
where $\|.\|_{\mathrm{\mathrm{TV}}}$ is the total variation distance, using $\Tilde{O}\left(\beta^{2}d^{3/2} \epsilon^{-2} \rho^{-2} c_{\mathrm{LSI}}^{-1} \right)$ queries to $\mathcal{O}_{\tilde{\nabla} f}$ and each $\mathcal{O}_{\tilde{\nabla} f}$ involves $O(d)$ gradient calculations.
\end{restatable}
As a byproduct, we present a quantum algorithm for estimating the partition functions of such non-logconcave distributions. The following theorem states the query complexities for different oracles, which is proved in \cref{sec:partitionfunction}.
\begin{restatable}{theorem}{partitionfunction}
    \label{th:partition}
    Let $Z = \int_x e^{-f(x)}\,\dd x$ be the partition with $f(x)$ function satisfying assumptions \cref{assumption1,assumption2}. Then, there exists quantum algorithms that output an estimate $\Tilde{Z}$ such that,
    \begin{equation}
        (1-\epsilon)Z\leq \Tilde{Z} \leq (1+\epsilon)Z
    \end{equation}
    with probability at least $3/4$ using,
    \begin{itemize}
      \item  $\Tilde{O}\left(d^{3/2}\epsilon^{-1} \rho^{-1} c_{\mathrm{LSI}}^{-1}\right)$ queries to $\mathcal{O}_{\nabla f}$ and $\mathcal{O}_f$, or 
      \item  $\Tilde{O}\left(d^{2}\epsilon^{-2}\rho^{-1} c_{\mathrm{LSI}}^{-1}\right)$ queries to $\mathcal{O}_{\nabla f}$, or 
      \item  $\Tilde{O}\left(d^{3}\epsilon^{-3}\rho^{-2} c_{\mathrm{LSI}}^{-1}\right)$ queries to $\mathcal{O}_{\tilde{\nabla} f}$.
    \end{itemize} 
\end{restatable}

\begin{remark}
    We note that under \cref{assumption1,assumption2}, Gibbs measure satisfies Log-Sobolev inequality with a finite constant by proposition 3.2 in \cite{1702.03849}. Therefore, we did not add Log-Sobolev inequality as an additional assumption (See \cref{remark:constants}). 
\end{remark}
\paragraph{Related work}
The convergence behavior of unadjusted Langevin algorithm is recently analyzed in \cite{vempala2022rapid}, \cite{ma2019sampling}, \cite{1507.05021}. Although the gradient descent methods are known to be superior to sampling-based optimization in convex cases, it was shown in \cite{ma2019sampling}, sampling-based optimization methods could be better in non-convex cases where the local information is not sufficient to find the global minimum. The stochastic extension of the algorithm is also analyzed recently. \cite{1702.05575} analyzed the hitting time of the stochastic Langevin dynamics to a neighborhood of the minima and \cite{2010.09597} analyzed the mixing time of SGLD to the stationary distribution in total variation distance. The comparison is summarized in \cref{tab:comparison}.

\begin{table}[h!]
\begin{adjustbox}{width=0.95\textwidth, center}
\begin{tabular}{ |c|c|c|c| } 
\hline
 Algorithm & Query Complexity  & Oracle& Assumptions\\
\hline
 ULA \cite{ma2019sampling}& $\tilde{O}(d/\epsilon^2)$& Full Gradient& Local non-convex \\
 MALA \cite{ma2019sampling}& $\tilde{O}(d^{2})$& Full Gradient and Evaluation & Local non-convex \\ 
 ULA \cite{vempala2022rapid}& $\tilde{O}(d/\epsilon^2)$ & Full Gradient& Isoperimetry  \\ 
 SGLD \cite{2010.09597}& $\tilde{O}(d^4/\epsilon^2)$& Stochastic Gradient&Dissipative Gradients \\
\hline
  Quantum MALA (\cref{thm:quantum-mala}) & $\tilde{O}(d )$  & Full Gradient and Evaluation&Dissipative Gradients \\ 
  Quantum ULA (\cref{thm:quantum-ula}) & $\tilde{O}(d^{3/2}/\epsilon)$ & Full Gradient& Dissipative Gradients\\ 
  Stochastic Quantum ULA (\cref{thm:quantum-sula}) & $\tilde{O}(d^{5/2}/\epsilon^2)$\tablefootnote{Note that in \cref{thm:quantum-sula}, the number of queries to $\mathcal{O}_{\tilde{\nabla} f}$ scales as $d^{3/2}$. Because each $\mathcal{O}_{\tilde{\nabla} f}$ uses $\tilde{O}(d)$ gradient calculations, the total gradient calculations scale as $d^{5/2}$.}& Stochastic Gradient& Dissipative Gradients\\
\hline
\end{tabular}
\end{adjustbox}
\caption{The comparison of our sampling algorithm with best-known classical results. Here we focus the dependencies on $d$ and $\epsilon$.\label{tab:comparison}}
\end{table}

The quantum random walk operator is based on \cite{1366222}. It is not easy to speed up the mixing time of a general random walk. However \cite{wocjan2008speedup} showed that for slowly varying Markov chains, it is possible to achieve quadratic speed up if the overlap between the stationary distributions of Markov chains has a large overlap and the initial distribution could be prepared efficiently. This technique has recently been used in quantum optimization problems such as estimating the partition function (\cite{montanaro2015quantum}), or volume estimation of convex bodies (\cite{CCH+19}). We summarize the classical results that use similar assumptions.

\begin{remark}
    \label{remark:constants}
    Even though our results imply that our quantum algorithm might have better dependency on Cheeger and Log Sobolev constants in stochastic case, we intentionally omit these terms in the table while making a comparison. This is because these constants depend on the function landscape and might have different dependence on problem dimension depending on the underlying assumption. For instance, assuming that the function is locally non-convex gives a dimension independent $c_{\mathrm{LSI}}$, whereas in general they might have exponential dependence. See \cite{2010.09597} for a more detailed discussion on these constants.
\end{remark}

\section{Preliminaries}
\paragraph{Notation}
The notation $\|\cdot\|$ denotes the spectral norm for operators and the $l_2$ norm for quantum states.
For Markov chains, we use the notation $P(x,.)$ to denote the transition probability distribution for point $x\in \Omega$, whereas we use $p_{xy}$ to denote the probability of transitioning from point $x$ to $y$.
For a distribution $p(x)$ and a function $q(x)$, the notation $p(x)\propto q(x)$ means $p(x)/p(y) = q(x)/q(y)$ and occasionally used for conciseness to hide the normalization factors. The ket notation $\ket{\nu}$ is sometimes referred to the coherent quantum state corresponding to probability distribution $\nu$ and is not explicitly stated when it is clear from the context.  
For distributions, $\|\cdot \|_{\mathrm{TV}}$ represents the total variation distance and $\|\cdot\|_H$ Hellinger distance.
The total variation distance between two probability distributions $P$ and $Q$ on $\Omega$ is defined as:
\begin{align}
   \|P-Q\|_{\mathrm{TV}} \coloneqq \sup_{A\subseteq \Omega} |P(A)-Q(A)|, 
\end{align}
and the Hellinger distance is defined as,
\begin{align}
    \|P-Q\|_{\mathrm{H}} = \left(\frac{1}{2}\int\limits_{x\in \Omega}\left(\sqrt{P(\dd x)}-\sqrt{Q(\dd x)}\right)^2  \right)^{1/2},
\end{align}
and finally, Wasserstein-2 distance is defined as, 
\begin{equation}
    \mathrm{W}_2(P, Q) \coloneqq \left(\inf_{z\in \Gamma (P,Q)}\int_{\mathbb{R}^d \times\mathbb{R}^d } \|x-y\|^2\,\dd z(x,y)  \right)^{1/2}
\end{equation}
where $\Gamma$ is the set of all couplings between $P$ and $Q$.

\subsection{Classical MCMC}
Monte Carlo Markov Chain (MCMC) methods are powerful computational techniques used for simulating and exploring complex probabilistic systems. In the context of sampling, MCMC involves constructing a Markov chain over the distribution's state space $\Omega$, where each state represents a potential sample. By iteratively transitioning between states according to carefully designed transition probabilities, MCMC methods generate a sequence of samples that converges to the target distribution.

Let \(P\) be the transition matrix of a Markov chain over a finite state space \(\Omega\), and let \(\pi\) be the stationary distribution of this chain. A stationary distribution $\pi$ is a probability distribution over states that remain unchanged under the transition dynamics of the Markov chain. Mathematically, it satisfies the balance equation, $\pi^TP  = \pi^T$. Hence $\pi$ is an eigenvector of $P$ with eigenvalue 1. The mixing time of a Markov chain can be defined as follows:

For any \(0 < \epsilon < 1\), the mixing time \(t_{\text{mix}}(\epsilon)\) is the smallest positive integer \(t\) such that for all initial distributions \(x\) in \(\Omega\):
\begin{align}
    \| P^t(x, \cdot) - \pi \|_{\mathrm{TV}} \leq \epsilon. 
\end{align} 

Here, \(P^t(x, \cdot)\) is the distribution of states after \(t\) steps starting from initial distribution \(x\), and \(\| \cdot \|_{\mathrm{TV}}\) represents the total variation distance between two probability distributions. The mixing time \(t_{\text{mix}}(\epsilon)\) characterizes the rate at which the chain approaches its stationary distribution within a specified tolerance \(\epsilon\).  Bounding the mixing time of the Markov chain is often a primary obstacle when it comes to proving the total algorithm's running time. Ergodicity and reversibility are crucial properties that significantly simplify the estimation of mixing time in Markov chains. A chain is said to be ergodic if it is irreducible (any state can be reached from any other state) and aperiodic (states can be revisited after a finite number of steps). A markov chain that satisfies the detailed balance condition is known as time-reversible and it is a fundamental property for establishing the mixing time in terms of spectral gap or conductance. Mathematically, it can be expressed as: 
\begin{align}
   P_{xy}\pi_y = P_{yx}\pi_x. 
\end{align}

This condition also guarantees that the Markov chain will converge to a stationary distribution. Furthermore, for a reversible Markov chain, it holds that
\begin{align}
    t_{\text{mix}}(\epsilon) \leq \frac{1}{2\gamma} \cdot \log\left(\frac{1}{\epsilon}\right).
\end{align}
$\gamma$ is the spectral gap defined as the difference between the first and the second-largest eigenvalue (in absolute value) of the transition matrix.

\subsection{Langevin Diffusion}
Langevin diffusion, often referred to as Langevin dynamics or Langevin Monte Carlo, is a fundamental stochastic differential equation that describes the dynamics of a particle undergoing random motion in a fluid or a complex environment. It is widely used in various scientific disciplines, including physics, chemistry, and biology, to model systems exhibiting Brownian motion or other forms of random behavior. It can be expressed as  a continuous-time stochastic process $X_t$ in the following form:
\begin{equation}
    \label{eq:langevin-diffusion}
    dX_t = -\nabla f(X_t)\,\dd t + \sqrt{2} \dd W_t, 
\end{equation}
where $W_t$ is the standard Brownian motion.

Langevin dynamics can also be used to explore high-dimensional parameter spaces and estimate probability distributions. In machine learning, one common goal is to optimize a loss function to train a model. Traditional gradient-based optimization techniques, like stochastic gradient descent, often struggle in complex, multimodal, or high-dimensional spaces, as they can get stuck in poor local minima. Langevin diffusion provides a probabilistic approach to optimization by simulating the motion of particles under the influence of both deterministic gradient forces and random noise. This allows the optimization process to explore the parameter space more extensively, potentially escaping local optima and reaching a broader range of solutions. By simulating Langevin dynamics, machine learning practitioners can sample from the posterior distribution of the model parameters, enabling Bayesian inference and uncertainty estimation. This makes Langevin diffusion particularly useful for Bayesian neural networks, where it provides a mechanism for approximating the posterior distribution over weights.
\subsection{Unadjusted Langevin Algorithm (ULA)}
Under certain conditions on $f$, \cref{eq:langevin-diffusion} accepts $e^{-f(x)}$ as its stationary density. Therefore, it is natural to discretize the Langevin diffusion using step size $\eta >0$. The simplest discretization scheme is known as Euler-Maruyama method and it gives the following update rule:
\begin{equation}
\label{eq:langevin}
    \x_{k+1} = \x_k - \eta \nabla f(\x_k) + \sqrt{2\eta\beta^{-1}}\z_k,
\end{equation}
where $\z_{\{0,1,...\}}$ are i.i.d Gaussian random vectors in $\bbr ^d$. 

\begin{algorithm}
\begin{algorithmic}
  \caption{Unadjusted Langevin Algorithm (ULA)\label{alg:ULA}}
\State \textbf{Input:} $\x_0, (\eta>0)$
\State \textbf{Output:} $\x_{K}$
\For{$k=1,...,K$}
\State $\x_{k} = \x_{k-1} - \eta \nabla f(\x_{k-1}) + \sqrt{2\eta/\beta}\z_{k-1}$
\EndFor
\State Return $\x_K$
\end{algorithmic}
\end{algorithm}

\subsection{Metropolis Adjusted Langevin Algorithm (MALA)}
Although ULA algorithm seems appealing due to its simplicity, it comes with a catch. Due to naive discretization of a continuous differential equation, the chain is asymptotically biased. That is, its stationary distribution is different than the stationary distribution of Langevin equation where the size of discrepancy depends on step size, feature dimension and properties of $f$. Furthermore, the chain does not satisfy the detailed balance condition, which is the standard assumption in the mixing time analysis of Markov chains.
A common practice is to apply Metropolis-Hasting correction at the end of each step to make the Markov chain reversible. This way, the chain converges to the desired target state and becomes time reversible. The update of the algorithm is modified such that if the following condition is satisfied, the iterates stay the same rather than applying \cref{eq:langevin}.
\begin{equation}
    \frac{p(\x_k|\x_{k+1})p^{\star}(\x_k)}{p(\x_{k+1}|\x_k)p^{\star}(\x_{k+1})}<u,
\end{equation}
where $u\sim \mathcal{U}[0,1]$. This algorithm is called Metropolis adjusted Langevin algorithm (MALA).
\begin{algorithm}
\begin{algorithmic}
  \caption{Metropolis Adjusted Langevin Algorithm (MALA)\label{alg:MALA}}
\State \textbf{Input:} $\x_0, (\eta>0) $
\State \textbf{Output:} $\x_{K-1}$
\For{$k=1,...,K$}
\State $\x_{k} = \x_{k-1} - \eta \nabla f(\x_{k-1}) + \sqrt{2\eta\beta^{-1}}\z_{k-1}$
\State $\alpha = \frac{p(\x_{k-1}|\x_{k})p^{\star}(\x_{k-1})}{p(\x_{k}|\x_{k-1})p^{\star}(\x_{k})}$
\State $u \sim \mathcal{U}[0,1]$
\If {$\alpha<u$}
\State $\x_{k} = \x_{k-1}$
\EndIf
\EndFor
\State Return $\x_K$
\end{algorithmic}
\end{algorithm}

However, in many optimization problems, computing the gradient and applying Metropolis step can become highly costly in terms of computation time. For example, in large-scale machine learning the objection function consists of a sum with a large number of terms and MALA requires $\Omega(N)$ function evaluations for data size $N$. Motivated by the large-scale optimization problems, we focus on unadjusted version of the algorithm with stochastic gradients and we believe it is worthwhile to study the possible speed-up that could be achieved on a quantum computer.

\subsection{Quantum Sampling}
Quantum sampling problem is to create the coherent version of the desired probability distribution:
\begin{equation}
    \label{eq: coherent_state}
    \ket{\pi} = \int \limits_{\x\in \bbr^d}\,\dd x\, \sqrt{\pi(x)}\ket{x},
\end{equation}
Then a measurement on this state yields the basis $\ket{x'}$ with probability $\pi(x')$. Let $\pi^{1/2}$ denote the diagonal matrix with entries $\sqrt{\pi(x)}$. For a time reversible Markov chain,
\begin{equation}
  D(P) = \pi^{1/2} P \pi^{-1/2}.  
\end{equation}
This follows from the detailed balance condition as follows:
\begin{align}
    D_{xy} &= \sum_{z,l}\pi^{1/2}_{xz} P_{zl}\pi^{-1/2}_{ly} \\
    &= \sqrt{\pi(x)}P_{xy}\sqrt{\pi(y)}\\
    &= \sqrt{P_{xy}P_{yx}}.
\end{align}
Therefore, the spectrum of $P$ matches the spectrum of $D$. Furthermore, 
\begin{align}
   D(P)\ket{\pi} = \ket{\pi}. 
\end{align}
Therefore the state $\ket{\pi}$ is an eigenvector of $D(P)$ with eigenvalue 1. The state in \cref{eq: coherent_state} can be prepared by using fixed point quantum amplitude amplification techniques as described in the previous section. 

\paragraph{Quantum walk}
Classical Markov chains can be quantized on a quantum computer using Szgedy's quantum walk operators introduced in \cite{1366222} by constructing a unitary operator on $\mathcal{H} = \mathbb{C}^N \otimes \mathbb{C}^N$. The action of quantum walk operator $U$ is defined as:
\begin{equation}
    U\ket{x}\ket{0} = \ket{\psi_x}
\end{equation}
where for all $x\in \Omega$, $\ket{\psi_x}$ is the state defined by,
\begin{equation}
    \ket{\psi_x} = \sum_{y\in \Omega}  \sqrt{p_{xy}}\ket{x}\ket{y}.
\end{equation}
The unitary operator $U$ can be realized as:
\begin{equation}
    U \coloneqq S(2\sum_x \ket{\psi_x}\bra{\psi_x}   - I),
\end{equation}
where $S = \sum_{xy}\ket{x}\ket{y}\bra{x}\bra{y}$ is the swap operator. The spectrum of $U$ and the singular value decomposition of the discriminant matrix $D_{xy} \coloneqq \sqrt{p_{xy}p_{yx}}$ are closely related. The spectrum theorem \cite[Theorem 1]{1366222} is the theoretical foundation of the quantum walk related algorithms.

The high-level idea of the sampling algorithm is to first construct a projection operator $V$ such that 
\begin{align}
    V\ket{\psi_0} = \ket{\psi_0}, 
\end{align}
and 
\begin{align}
    V\ket{\psi_{j\neq 0}} = 0,
\end{align}
where $\ket{\psi_j}$ is the eigenvector of $U$ with eigenvalue $e^{2\pi i \theta_j}$. This projection operator can be approximately built by using a quantum phase estimation circuit or quantum singular value transformation(\cite{Gily_n_2019}) to eliminate the phases smaller than $\theta_2$. Then one can drive an initial quantum state $\ket{\mu_0}$ to $\ket{\pi}$ by repeatedly applying this projection operator.

\paragraph{Implementing quantum walk operators}
We use stochastic gradient oracle $O_{\Tilde{\nabla}f}$ to prepare the following state,
\begin{equation}
  \label{eq:gradient-mapping}
  \ket{x}\ket{0} \mapsto \ket{x}\ket{\Tilde{\nabla} f(x)},
\end{equation}
Then the one step of the walk can be implemented by shifting the register that holds the gradient with Gaussian state. A Gaussian state can be prepared using Box-Muller transformation.
\begin{equation}
  \ket{x}\ket{0} \mapsto \ket{x}\int_{\bbr^d}\dd y\, \sqrt{p_{xy}} \ket{y},
\end{equation}
where
\begin{align}
  p_{xy} = \left(\frac{1}{2\pi}\right)^{d/2}e^{-\frac{1}{2}\norm{y-\eta \Tilde{\nabla} f(x)}_2^2}.
\end{align}
The query complexity of this operation is $O(B)$ due to $B$ gradient evaluations required to implement \cref{eq:gradient-mapping}.

Note that the we present the quantum states and operators in the continuous-space representation as introduced in~\cite{CCH+19}. The analysis in continuous-space simplifies the analysis, while the implementations are always in a discretized space (as we only have finite bits of precision for real numbers). We refer to~\cite{CCH+19} for the error analysis caused by the discretization, which is not dominating other errors.

\paragraph{Implementing projection operators}
To implement amplitude amplification, one needs the projection operator $V_{\pi}$ that projects any quantum state onto a subspace spanned by the target state $\pi$. Let $V_{<\gamma}$ project any quantum state to the eigenvector of the quantum operator $U$ with eigenphase smaller than $\gamma$. This operator can be implemented by using quantum phase estimation or QSVT techniques by using $\Tilde{O}(1/\gamma)$ calls to controlled $U$ operators with $\epsilon$ accuracy. If $U$ is a quantum walk corresponding to quantum MALA algorithm, by setting $\gamma = \Delta$ where $\Delta$ is the phase gap of $U$, this projector projects any state onto desired Gibbs state $\ket{\pi}$. However, in quantum ULA or stochastic ULA, this projector has a bias. That is, it only projects the state onto a neighborhood of $\ket{\pi}$. Our main contribution is to quantify this discrepancy and bound the step size for sufficiently small error.

\paragraph{Amplitude amplification}
Once we have the projector operators, we use the $\pi/3$ amplitude amplification technique introduced in \cite{Grover_2005} to drive the initial state to the target state by applying reflection operators iteratively.

\subsection{Other Technical Lemmas}
The following lemma characterizes the conductance parameter of the classical MALA algorithm constructed with stochastic gradients under given assumptions. Though similar results are given for full gradient case in \cite{ma2019sampling}, we use the stochastic version and remove $B$ dependent condition on the step size when we apply this lemma in full gradient case by setting $B \gg d$.
\begin{lemma}[Lemma 6.5 in \cite{2010.09597}]
\label{lemma13}
Under \cref{assumption1,assumption2}, if the step size meets the condition $\eta \leq \min\left\{35 (Ld + (LR+G)^2\beta d/B)]^{-1}, [25\beta (LR+G)^2]^{-1}\right\} $, then there exists absolute constant $c_0$ such that, the conductance parameter $\phi$ for Metropolis adjusted Stochastic Langevin Algorithm satisfies,
    \begin{equation}
        \phi\geq c_0 \rho \sqrt{\eta/\beta}.
    \end{equation}
    where $\rho$ is the Cheeger constant of the truncated distribution $\pi^{\star}$.
\end{lemma}

The next two lemmas are used in our proofs commonly. The first one lower bounds $f(x)$ by a quadratic function whereas the second one upper bounds the norm of the gradient by a linear function.
\begin{lemma}[Lemma A.1 in \cite{2010.09597}]
\label{lemma:lower-bound-on-f}
    Under \cref{assumption2}, the objective function $f(x)$ satisfies,
    \begin{equation}
        f(x)\geq \frac{m}{4}\|x\|^2 + f(x^{\star})-b/2.
    \end{equation}
\end{lemma}

\begin{lemma}[Lemma 3.1 in \cite{1702.03849}]
\label{lemma:raginsky}
    Under \cref{assumption1}, there exists a constant $G= \max_{k\in [N]}\|\nabla f_k(0)\|$ such that for any $x\in \mathbb{R}^d$ and $k\in [n]$, it holds that,
    \begin{equation}
        \|\nabla f_k(x)\|\leq L\|x\|+G.
    \end{equation}
\end{lemma}

\section{Annealing Schedule for Non-Logconcave Distributions}
To implement slowly varying Markov chains for non-logconcave distributions, we prove the following lemma to define the annealing schedule.
\begin{lemma}
\label{lemma:annealing}
Under \cref{assumption1,assumption2}, there exists a series of quantum states $\ket{\mu_0}, \ket{\mu_1}, \ldots, \ket{\mu_{M-1}}$ satisfying the following properties: 
\begin{enumerate}
    \item There exists an efficient quantum algorithm to prepare $\ket{\mu_0}$.
    \item For all $i \in [0, M-1]$ , $\ket{\mu_i}$ and $\ket{\mu_{i+1}}$ has at least constant overlap, i.e.,
    \begin{equation}
        |\bra{\mu_{i}}\ket{\mu_{i+1}}|\geq \Omega(1).
    \end{equation}
     \item The final state $\ket{\mu_{M-1}}$ has at least constant overlap with the target Gibbs state $\ket{\pi}$,
     \begin{equation}
        |\bra{\mu_{M-1}}\ket{\pi}|\geq \Omega(1).
    \end{equation}
    \item The number of quantum states $M\leq \Tilde{O}(c_{\mathrm{LSI}}^{-1}d^{1/2})$.
\end{enumerate}
\end{lemma}
\begin{proof}
Our construction and analysis are similar to the annealing scheme used in \cite{childs2022quantum}, however our proof does not require any convexity assumption for $f(x)$. The construction is as follows:
\begin{enumerate}
    \item $\ket{\mu_0} = \sum\limits_{x} \sqrt{p_0(x)}\ket{x}$, where $p_0(x) = \frac{\exp(-\frac{\|x\|^2}{2\sigma_1^2})} {Z_0}$.
    \item For all $i\in [1, M-1]$, $\ket{\mu_i} = \sum\limits_{x\in \Omega} \sqrt{p_i(x)}\ket{x}$, where $p_i(x) =\frac{\exp(-f(x)-\frac{\|x\|^2}{2\sigma_i^2})}{Z_i}$ such that $\sigma_{i+1}^2 = \sigma_i^2(1+\alpha)$ with $\alpha = \Tilde{O}(d^{-1/2}c_{\mathrm{LSI}})$.
\end{enumerate}
Here, $Z_0 = \int \dd x\, \exp(-\frac{\|x\|^2}{2\sigma_1^2})$ and $Z_i = \int \dd\, x \exp(-f(x)-\frac{\|x\|^2}{2\sigma_i^2})$.
The first property in the lemma statement holds, since $p_0$ corresponds to a Gaussian distribution and the coherent quantum state corresponding to Gaussian distributions can be efficiently prepared by using Box-Muller technique without using any evaluation of $f$ or $\nabla f$. Next, we prove the second property. We first start with $i=0$ as the base case: $  \bra{\mu_0}\ket{\mu_1} \geq \Omega(1)$. To prove this, Let $f(x^{\star}) = \min\limits_{x\in \Omega}f(x)$. We fix $\beta = 1$ without loss of generality. Then, we can write,
\begin{equation}
\label{eq:upper-bound-on-f}
    f(x) \leq f(x^{\star}) + \langle \nabla f(x^{\star}), x - x^{\star} \rangle +\frac{L}{2}\|x-x^{\star}\|^2\leq f(x^{\star}) + L\|x^{\star}\|^2+L\|x\|^2.
\end{equation}
 where the first inequality is well known due to \cref{assumption1} (see \cite{nesterov2018lectures}) and second inequality is due to Young's inequality.
Using this upper bound on $f(x)$, we have,
\begin{align}
    \bra{\mu_0}\ket{\mu_1}  &=  \frac{\int \,\dd x \exp(-\frac{1}{2}f(x)-\frac{\|x\|^2}{2\sigma_1^2}) }{(2\pi \sigma_1^2)^{d/4}\sqrt{Z_1}}\\
    &\geq \frac{\int \,\dd x \exp(-\frac{1}{2}f(x^{\star})-\frac{1}{2}L\|x\|^2 - \frac{1}{2}L\|x^{\star}\|^2-\frac{\|x\|^2}{2\sigma_1^2})}{(2\pi \sigma_1^2)^{d/4}\sqrt{\int \,\dd x \exp(-f(x^{\star}) -\frac{\|x\|^2}{4\sigma_1^2})}}\\
    &=\frac{\exp(-\frac{L}{2}\|x^{\star}\|^2)}{(2\pi \sigma_1^2)^{d/4}} \frac{\pi^{d/2}(L/2+1/(2\sigma_1^2))^{-d/2}}{(2\pi \sigma_1^2)^{d/4} }\\
    &=\exp(-\frac{L}{2}\|x^{\star}\|^2) (L\sigma_1^2+1)^{-d/2}\\
    &\geq \exp(- \frac{L}{2}\|x^{\star}\|^2-\frac{dL\sigma_1^2}{2}). 
\end{align}
Choosing $\sigma_1^2 = \frac{\epsilon}{2dL}$ yields  $|\bra{\mu_0}\ket{\mu_1} |\geq \Omega(1)$. Next, we consider $1\leq i\leq M-1$. Letting $\sigma^2 = \sigma_{i+1}^2$, we have
\begin{align}
    |\bra{\mu_i}\ket{\mu_{i+1}}|&= \int \,\dd x\frac{\exp(-f_i(x)/2) }{\sqrt{Z_i}}  \frac{\exp(-f_{i+1}(x)/2) }{\sqrt{Z_{i+1}}} \\
    &= \int \,\dd x \frac{\exp(-f(x) - \frac{\|x\|^2}{4\sigma_i^2} - \frac{\|x\|^2}{4\sigma_{i+1}^2}) }{\sqrt{Z_i Z_{i+1}}}\\
    &=\frac{\mathbb{E}_{\pi}\left[\exp(-\frac{1+\alpha/2}{2\sigma^2}\|x\|^2 ) \right]}{\mathbb{E}_{\pi}\left[\exp(\frac{-1+\alpha}{2\sigma^2}\|x\|^2 )\right]  ^{1/2}\mathbb{E}_{\pi}\left[\exp( \frac{-1}{2\sigma^2}\|x\|^2 )\right]^{1/2} },
\end{align}
where the last step follows from the fact that the numerator can be written as,
\begin{align}
    \int \,\dd x \exp(-f(x) - \frac{\|x\|^2}{4\sigma_i^2} -\frac{\|x\|^2}{4\sigma_{i+1}^2})  =  \frac{Z\int \,\dd x \exp(-f(x) -  \frac{1+\alpha/2}{2\sigma^2} \|x\|^2) }{Z} = Z\mathbb{E}_{\pi}\left[\exp(-\frac{1+\alpha/2}{2\sigma^2}\|x\|^2) \right],
\end{align}
and similarly, $Z_i$ and $Z_{i+1}$ can be simplified as,
\begin{align}   
&Z_i = \int \,\dd x e^{-f(x) - \frac{1}{2\sigma_i^2}\|x\|^2} = \frac{Z\int \,\dd x\exp({- f(x)- \frac{(1+\alpha)}{2\sigma^2}\|x\|^2  })}{Z} = Z\mathbb{E}_{\pi} \left[\exp({-\frac{(1+\alpha)}{2\sigma^2} \|x\|^2 })\right] \\
&Z_{i+1} = \int \,\dd x e^{- f(x) - \frac{1}{2\sigma_{i+1}^2}\|x\|^2} = \frac{Z\int\exp({- f(x)-\frac{1}{2\sigma^2} \|x\|^2 })}{Z} = Z\mathbb{E}_{\pi} \left[\exp({-\frac{\|x\|^2}{2\sigma^2}  })\right].
\end{align}

Defining $\alpha' = \frac{\alpha}{\alpha+2}$ and $\sigma'^2 =\frac{\sigma^2}{1+\alpha/2}$, we have
\begin{align}
     |\bra{\mu_i}\ket{\mu_{i+1}}| &= \frac{\mathbb{E}_{\pi}\left[\exp(\frac{-1}{2\sigma'^2}\|x\|^2 ) \right]}{\mathbb{E}_{\pi}\left[\exp(-\frac{1+\alpha'}{2\sigma'^2}\|x\|^2 )\right]^{1/2}  \mathbb{E}_{\pi}\left[\exp(-\frac{1-\alpha'}{2\sigma'^2}\|x\|^2 )\right]^{1/2} } \\
     & \geq \Omega(\exp(-2dL\alpha'^2/(mc_{\mathrm{LSI}}^2))).
\end{align}
where the last inequality is due to $\cref{lemma11}$. Setting $\alpha^2 = \Tilde{O}(c_{\mathrm{LSI}}^2m/(dL))$, we have $ |\bra{\mu_i}\ket{\mu_{i+1}}| \geq \Omega(1)$.
Having established property 2, we move on to the third property:
\begin{align}
     |\bra{\mu_{M-1}}\ket{\pi}| &= \int \,\dd x \frac{\exp(-f(x) -\frac{\|x\|^2}{4\sigma_{M-1}^2} )}{\sqrt{Z_{M-1}}\sqrt{Z}}\\
     &= \mathbb{E}_{\rho'}\left[\exp(-\frac{1}{4\sigma_{M-1}^2}\|x\|^2)  \right]^{-1/2}\mathbb{E}_{\rho'}\left[\exp(\frac{1}{4\sigma_{M-1}^2}\|x\|^2)  \right]^{-1/2} \\
     &\geq 1-\Omega(dL/(m\sigma_{M-1}^4c_{\mathrm{LSI}}^2))
\end{align}
where $\rho'\propto \pi(x)\exp(-\frac{\|x\|^2}{4\sigma_{M-1}^2})$. The last step is due to $\cref{lemma10}$. Setting $\sigma_M^2 =\sqrt{dL/(mc_{\mathrm{LSI}}^2)}$ satisfies $  |\bra{\mu_{M-1}}\ket{\pi}|\geq \Omega(1)$. The final property follows from the fact that $\alpha  = \Tilde{O}(\sqrt{c_{\mathrm{LSI}}^2m/(dL)})$, since,
\begin{align}
    \sigma_{M-1} = \sigma_0(1+\alpha)^{M-1}
\end{align}
and solving this for $M$ yields $M = \Tilde{O}(\sqrt{dL/(m c_{\mathrm{LSI}}^2)})$. 
\end{proof}

The next three technical lemmas are presented to make the proof of the annealing schedule concise. 
\begin{lemma}
\label{lemma10}
Suppose $\pi(x)\propto e^{-f(x)}$ is a Gibbs measure and $f$ satisfies \cref{assumption1,assumption2}. Then, we have
\begin{align}
     \mathbb{E}_{\pi}\left[\exp(-s\|x\|^2 )\right]\mathbb{E}_{\pi}\left[\exp(s\|x\|^2 )\right] \leq O(\exp(dLs^2/(m c_s^2)))
\end{align}
where $ c_{s}^2$ is Log-Sobolev constant of the distribution $\pi_s \propto \pi e^{s\|x\|^2} $.
\end{lemma}
\begin{proof}
Let $h(s) = \mathbb{E}_{\pi}\left[\exp(-s\|x\|^2 )\right]\mathbb{E}_{\pi}\left[\exp(s\|x\|^2 )\right] $, then
\begin{align}
    \frac{h'(s)}{h(s)}&= \left(\frac{\mathbb{E}_{\pi}\left[ \|x\|^2 \exp (s\|x\|^2) \right]}{\mathbb{E}_{\pi}\left[\exp(s\|x\|^2)\right] } - \frac{\mathbb{E}_{\pi}\left[\|x\|^2 \exp (-s\|x\|^2)\right] }{\mathbb{E}_{\pi}\left[\exp(-s\|x\|^2) \right]}\right) \\
    &=\int^{s}_{-s} v'(t),\dd t,
\end{align}
where $v(t)$ is defined as,
\begin{align}
    v(t) = \frac{\mathbb{E}_{\pi} \left[\|x\|^2 \exp(t \|x\|^2)\right]}{\mathbb{E}_{\pi}\left[ \exp(t \|x\|^2)\right]}.
\end{align}
Computing $v'(t)$ gives,
\begin{align}
    v'(t) & = \frac{\mathbb{E}_{\pi} \left[\|x\|^4 \exp(t\|x\|^2)\right]  \mathbb{E}_{\pi}\left[\exp(t\|x\|^2)\right] - (\mathbb{E}_{\pi} \left[ \|x\|^2 \exp(t\|x\|^2)\right]^2}{\left(\mathbb{E}_{\pi}\left[\exp(t\|x\|^2)\right]\right)^2}\\
    &= \mathrm{Var}_{\pi_t}\|x\|^2,
\end{align}
where $\pi_t$ is a distribution defined as,
\begin{align}
    \pi_t(x) \propto \pi(x) \exp(t\|x\|^2 )
\end{align}
Suppose $\pi$ satisfies Log-Sobolev inequality with constant $c_{\mathrm{LSI}}$, it also satisfies the Poincare inequality with the same constant (e.g \cite{GOEL200451}).
\begin{align}
  \mathrm{Var}_{\pi_t}[\|x\|^2]\leq \frac{1}{c_{t}}\mathbb{E}_{\pi_t}[\|x\|^2]\leq O(Ld/(mc_t^2)),
\end{align}
where $c_t$ is LSI constant of $\pi_t$ and the second inequality is due to \cref{lemma12}. Therefore,
\begin{align}
    \frac{h'(s)}{h(s)}=\int^{s}_{-s} v'(t)\,\dd t = O(dLs/(mc_s^2)).
\end{align}
Hence,
\begin{align}
   \log(h(s))-\log(h(0))=  \int_0^s   \frac{h'(t)}{h(t)}\, \dd t = O(dLs^2/(mc_s^2)).
\end{align}
Since $h(0)=1$, we conclude the proof.
\end{proof}

\begin{lemma}
\label{lemma11}
Suppose $\pi(x)\propto e^{-f(x)}$ is a Gibbs measure and $f$ satisfies the Log-Sobolev inequality with constant $c_{\mathrm{LSI}}$. Then under \cref{assumption1,assumption2},
\begin{equation}
     \frac{\mathbb{E}_{\pi}\left[\exp(-(1+\alpha)\|x\|^2 )\right]\mathbb{E}_{\pi}\left[\exp(-(1-\alpha)\|x\|^2 )\right]}{(\mathbb{E}_{\pi}[\exp(-\|x\|^2)])^2 } \leq O(\exp(dL\alpha^2/(c_{\mathrm{LSI}}^2m)))
\end{equation} 
for $0\leq \alpha\leq 1/2$.
\end{lemma}
\begin{proof}
    This follows from \cref{lemma10}, by setting $\Tilde{\pi} \propto \pi \exp(-\|x\|^2)$. Then,
    \begin{align}
        \frac{\mathbb{E}_{\pi}\left[\exp(-(1+\alpha)\|x\|^2 )\right]\mathbb{E}_{\pi}\left[\exp(-(1-\alpha)\|x\|^2 )\right]}{(\mathbb{E}_{\pi}[\exp(-\|x\|^2)])^2 } &=\mathbb{E}_{\Tilde{\pi}}\left[\exp(-\alpha\|x\|^2 )\right]\mathbb{E}_{\Tilde{\pi}}\left[\exp(\alpha \|x\|^2 )\right] \\
         &\leq O(\exp(dL\alpha^2/(mc_{\alpha}^2)))\\
         &\leq O(\exp(dL\alpha^2/(mc_{\mathrm{LSI}}^2)))
    \end{align}
    with $c_\alpha$ is LSI constant of $\pi_{\alpha}  \propto \pi \exp(-(1-\alpha)\|x\|^2)$. The last step follows from the fact that $c_{\alpha}\geq c_{\mathrm{LSI}}$ for $\alpha\leq 1/2$.
\end{proof}

\begin{lemma}
\label{lemma12}
Suppose $\pi(x)\propto e^{-f(x)}$ is a Gibbs measure and $f$ satisfies \cref{assumption1,assumption2}. Then
    \begin{equation}
        \mathbb{E}_{\pi_s(x)}[e^{s\|x\|^2}\|x\|^2]\leq O(Ld/(mc_s))
    \end{equation}
   where $\pi_s$ is a probability distribution proportional to $\pi(x)e^{s\|x\|^2}$ for a constant $s\leq \frac{m}{8}$ and $c_s$ is the Log-Sobolev constant of $\pi_s$.
\end{lemma}
\begin{proof}
    Our proof follows the idea presented in proof of Lemma 6 in \cite{ma2019sampling} without the assumption of local non-convexity.
    We choose an auxiliary random variable $x'$ following the law of $p \propto e^{-(L-s)\|x\|^2}$ and couples optimally with $x_s \sim \pi_s:(x_s,x')\sim \gamma \in\Gamma_{opt}(\pi_s,p)$.
\begin{align}
    \mathbb{E}_{\pi_s}\|x\|^2 &= \mathbb{E}_{(x_s,x'\sim \gamma)} \| x' -x' +x_s\|^2 \\
    &\leq 2\mathbb{E}_{p}\|x'\|^2 + 2\mathbb{E}_{(x_s,x'\sim \gamma)} \|x'-x_s\|^2\\
    &= \frac{2d}{L-s} + 2\mathrm{W}_2^2(p,\pi_s)\\
    &\leq \frac{2d}{L-s} + \frac{2}{c_{s}}\mathrm{KL}(p,\pi_s ),
\end{align}
where $c_{\pi_s}$ is LSI constant of $\pi_s$. The first inequality follows from Young's inequality and second inequality is due to generalized Talagrand inequality \cite{OTTO2000361}. KL divergence can be bounded,
\begin{align}
  \mathrm{KL}(p,\pi_s) &= \int_x \,\dd x\log(\frac{p(x)}{\pi_s(x)})p(x) \\
    &\leq \sup_x\log(\frac{p(x)}{\pi_s(x)})\int_x \dd x\, p(x)\\
    &= \sup_x \log(\frac{p(x)}{\pi_s(x)}).
\end{align}
We can further bound $\frac{p(x)}{\pi_s(x)}$ for any $x\in \Omega$,
\begin{align}
    \frac{p(x)}{\pi_s(x)} &= \frac{e^{-(L-s)\|x\|^2} }{\int_x \,\dd x\, e^{-(L-s)\|x\|^2 }}\frac{\int \,\dd x\, e^{-f(x)}e^{s\|x\|^2} }{e^{-f(x)}e^{s\|x\|^2}}\\
    &=\frac{\int \,\dd x\, e^{s\|x\|^2-f(x)}}{\int \,\dd x e^{-(L-s)\|x\|^2} }e^{-L\|x\|^2+f(x)} \\
    &\leq \frac{\int \,\dd x\, e^{s\|x\|^2-m\|x\|^2/4 +b/2 -f(x^{\star})} }{\int \,\dd x e^{-(L-s)\|x\|^2} }e^{-L\|x\|^2 +f(x)}\\
    &\leq \frac{\int \,\dd x\, e^{s\|x\|^2-m\|x\|^2/4 +b/2 -f(x^{\star})} }{\int \,\dd x e^{-(L-s)\|x\|^2} }e^{L\|x^{\star}\|^2 +f(x^{\star})}\\
    &= e^{b/2 + L\|x^{\star}\|^2} \frac{(L-s)^{d/2}}{(m/4-s)^{d/2}},
\end{align}
where the first inequality is due to \cref{assumption2} and \cref{lemma:lower-bound-on-f}. Second inequality follows from \cref{eq:upper-bound-on-f} Hence, KL divergence is bounded by,
\begin{align}
  \mathrm{KL}(p,\pi_s) &\leq   \sup_x \left(\frac{p(x)}{\pi_s(x)}\right) \leq  b/2 + L\|x^{\star}\| +\frac{d}{2}\log(\frac{L-s}{m/2-2s}).
\end{align}
This implies that, 
\begin{equation}
    \mathbb{E}_{\pi_s}\|x\|^2\leq \frac{2d}{L-s} + \frac{2}{c_s} (b/2 + L\|x^{\star}\|^2 +\frac{d}{2}\log(\frac{L-s}{m/2-2s}) =O(Ld/(mc_s))
\end{equation}
for $s\leq m/8$.
\end{proof}

\section{Quantum Algorithms for Sampling}

Let $P^{\star}(x\to y) = p^{\star}_{xy}$ and $P(x\to y) = p_{xy}$ denote the transition densities of MALA and ULA algorithms respectively. Similarly, let $U^{\star}$ and $U$ be the quantum walk operators associated with $P^{\star}$ and $P$ constructed.
The goal is to construct a quantum ULA algorithm that uses stochastic gradient to provide speedups with less expensive gradient oracles.
However, ULA is not a time reversible Markov chain due to the absence of Metropolis-Hastings step, and hence it is not known how to characterize its mixing time in terms of its spectral gap. As a result, it was not clear whether there is any quantum speedup using the existing techniques. Here, we first analyze quantum MALA algorithm and use it as a theoretical bridge to establish the complexity of our quantum ULA algorithm. That is, we first build a hypothetical quantum circuit that implements quantum MALA algorithm using $U^{\star}$ and establish its mixing time in terms of gate complexity. Since MALA is a reversible Markov chain, it is possible to connect its mixing time using a conductance analysis. Then, we replace each $U^{\star}$ with $U$ in the algorithm. This replacement will introduce an error in the final probability distribution as a function of learning rate. We obtain a bound on the learning rate to guarantee that the replacement of quantum gates introduces at most $\epsilon$ error. Finally, we use the bound on the learning rate to find the final complexity of our algorithm.

For technical reasons, we set the domain $\Omega = \mathbb{R}^d \cap B(0,R)$ where $R$ is sufficiently large enough to show that the truncated distribution $\pi^{\star}$ in $\Omega$ is $\epsilon$ close to the original Gibbs distribution. More specifically, we work on sufficiently large but bounded domain to show that the norm of the gradients are bounded and derive our results in terms of $R$. The truncation is done by only considering the sum of projectors up to $\|x\|\leq R$ in the implementation of the quantum walk. Then we use the following lemma to characterize $R$,
\begin{lemma}[Lemma 6 in \cite{2010.09597}]
\label{lemma:bounded}
    For any $\epsilon \in (0,1)$ set $R = \bar{R}(\epsilon/12) $ and let $\pi^{\star}$ be the truncated distribution in $\Omega$. Then the total variation distance between $\pi^{\star}$ and $\pi$ is upper bounded by $\|\pi^{\star}-\pi\|\leq \epsilon/4$, where 
    \begin{align}
        \bar{R}(z) =\left[\max\left\{\frac{625d\log(4/z)}{m\beta},\frac{4d\log(4L/m)}{m\beta}, \frac{4d+8\sqrt{d\log(1/z)}+8\log(1/z) }{m\beta}) \right\}\right]^{1/2}.
    \end{align}
\end{lemma}

\subsection{Quantum Metropolis Adjusted Langevin Algorithm}
\label{sec:q-mala}
The following lemma is useful to characterize the phase gap of quantum walk operator for a reversible Markov chain in terms of its conductance parameter and it is the source of the quantum speed up for mixing time for reversible chains. 

\begin{restatable}{lemma}{phasegap}
    \label{lemma:phase-gap}
    Let $Q$ be a reversible Markov chain with conductance parameter $\phi(Q)$ and the eigenvalues $\lambda_0=1> |\lambda_1|\geq |\lambda_2|\geq \cdots \geq |\lambda_m|$. Let $W$ be a unitary quantum walk operator constructed with the transition density of $Q$. Then the phase gap $\Delta(W) \coloneqq 2\arccos |\lambda_1|$ is lower  bounded by,
    \begin{equation}
        \Delta(W) \geq \sqrt{2}\phi(Q).
    \end{equation}
\end{restatable}
The proof is postponed in \cref{sec:prooflemma-qmala}. 
Now, we restate and prove our result for quantum MALA algorithm via next theorem.

\qmala*

\begin{proof}
Let $\ket{\mu_0},\ket{\mu_1},...,\ket{\mu_{M-1}}$ be the series of quantum states described in \cref{lemma:annealing}. 
We start with the preparation of the initial Gaussian state $\ket{\mu_0}$ which can be done efficiently by applying the Box-Muller transformation to the uniform distribution state (see \cite[Appendix A.3]{CCH+19} for more details). Then, for each $i\in [0,M-2]$, we drive each state $\ket{\mu_i}$ to $\ket{\mu_{i+1}}$ using $\pi/3$-fixed-point amplitude amplification algorithm (see \cite{Grover_2005}). The amplitude amplification uses the following reflection operators,
\begin{align}
    V_i &= e^{i\pi/3}\ket{\mu_i}\bra{\mu_{i}}+ (I-\ket{\mu_i}\bra{\mu_{i}}),\\
    V_{i+1}&=e^{i\pi/3}\ket{\mu_{i+1}}\bra{\mu_{i+1}}+ (I-\ket{\mu_{i+1}}\bra{\mu_{i+1}}).
\end{align}
Each state $\ket{\mu_i}$ is the unique eigenvector of quantum MALA operator $U_i^{\star}$ for $f_i(x) = f(x)+\frac{\|x\|^2}{2\sigma_i^2}$ since the classical MALA is time reversible and its stationary distribution is $\mu_i$. Therefore, the operator $\ket{\mu_i}\bra{\mu_{i}}$ is a projector operator to the eigenstate of $U_i^{\star}$ with eigenphase 0. Then, by \cite[Corollary 4.1]{CCH+19}, the operator $V_i$ can be implemented
with $\epsilon$ accuracy using $\Tilde{O}(1/\Delta(U_i^{\star}))$ calls to controlled-$U^{\star}$ operators where $\Delta(\cdot)$ is the phase gap. By \cref{lemma13}, \cref{lemma:bounded}, and \cref{lemma:phase-gap}, $\Delta \geq \rho\sqrt{2\eta/\beta}$ for step size smaller than $O(\min\{d^{-1},\beta^{-1} \})$. Then, using $\Tilde{O}(\rho^{-1}\eta^{-1/2}\beta)$ calls to controlled- $U_i^{\star}$ operators, we can implement a quantum reflection $\Tilde{V}_i$ such that,
\begin{equation}
   \|\Tilde{V}_i - V_i\|\leq \epsilon
\end{equation}
for each $i$. Then, given $\ket{\mu_i}$, we can drive $\ket{\mu_i}$ to $\ket{\Tilde{\mu}_{i+1}}$ using constant number of $\Tilde{V_i}$ operators because,  
\begin{align}
    |\bra{\mu_i}\ket{\Tilde{\mu}_{i+1}}|\geq \Omega(1).
\end{align}
such that $\|\ket{\Tilde{\mu}_{i+1}}-\ket{\mu_{i+1}}\|\leq \epsilon$.
Then we apply the same steps $M$ times to drive $\mu_0$ to $\pi$ with at most error $\epsilon$. Note that the error in each step does not accumulate linearly. This is because we can drive $\Tilde{\mu}_i$ to $\Tilde{\mu}_{i+1}$ with logarithmic cost in applying reflection operators. Since $M = \Tilde{O}(c_{\mathrm{LSI}}^{-1}\sqrt{d})$, the total complexity of the annealing procedure is $\Tilde{O}(d^{1/2}c_{\mathrm{LSI}}^{-1}\rho^{-1}\eta^{-1/2}\beta) =\Tilde{O}(c_{\mathrm{LSI}}^{-1}\rho^{-1}\beta d)$. Each $U^{\star}$ operator can be implemented using constant number of calls to full gradient and evaluation oracles, the algorithm uses $\Tilde{O}(c_{\mathrm{LSI}}^{-1}\rho^{-1}\beta d)$ full gradient and function evaluations.
\end{proof}

\subsection{Quantum Unadjusted Langevin Algorithm}
\label{sec:q-ula}
Since unadjusted Langevin algorithm is not a reversible chain, we cannot follow the same procedure since there is no direct relation between the conductance and phase gap. The following lemma, proved in \cref{sec:prooflemma-qula}, quantifies the error between two quantum walk operators corresponding to different Markov chains in spectral norm.

\begin{restatable}{lemma}{uutilde}
\label{lemma1}
Let $U$ and $\Tilde{U}$ be two quantum walk operators associated with two classical Markov chains with transition densities $P(x\to y) = p_{xy}$ and $\Tilde{P}(x\to y) = \Tilde{p}_{xy}$, respectively. Then,
\begin{equation}
\|U-\Tilde{U}\|\leq 4\sqrt{2}\max_x \|P(x, .)-\Tilde{P}(x,.)\|_{\mathrm{H}},    
\end{equation}
where $\|P(x, .)-\Tilde{P}(x,.)\|_{\mathrm{H}}$ denotes the Hellinger distance between the probability densities $P(x,.)$ and $\Tilde{P}(x,.)$ for any $x\in \Omega$.
\end{restatable}

To be able to apply \cref{lemma1}, we bound the Hellinger distance between the probability distributions of MALA and ULA algorithm through the next lemma, which is proved in \cref{sec:prooflemma-qula}.
\begin{restatable}{lemma}{pptilde}
\label{lemma2}
Let $P(x\to y) = p_{xy}$ and $P^{\star}(x\to y)= p^{\star}_{xy}$ be the transition densities for Unadjusted Langevin Algorithm (ULA) and Metropolis Adjusted Langevin Algorithm (MALA) respectively. Then under \cref{assumption1,assumption2}
and for step size $\eta\leq d(\beta (LR+G)^2)^{-1}$,
\begin{equation}
    \max_x \|P(x, .)-\Tilde{P}(x,.)\|_{\mathrm{H}}\leq 4\eta d L,
\end{equation}
where $G$ is a positive constant that satisfies $\|\nabla f(0)\|\leq G$.
\end{restatable}
As the quantum walk operator is the basic building block of the reflection operators used in amplitude amplification, we present the following result to relate the error in quantum walk operator to the projection operators. 

\begin{restatable}{lemma}{wwtilde}
\label{lemma3}
    Let $W$ be a unitary operator with phase gap $\Delta$ and assume that $W$ has a unique eigenvector $\ket{\psi_0}$ with eigenvalue 1. Suppose that we have $\Tilde{W}$ such that,
    \begin{equation}
        \|W - \Tilde{W}\|\leq \delta.
    \end{equation}
    Let $\Pi_{< \Delta}$ and $\Tilde{\Pi}_{< \Delta}$ be operators that project any quantum state onto the space of eigenvectors of $W$ and $\Tilde{W}$ with phases smaller than $\Delta$ respectively. Then,
    \begin{equation}
        \|\Pi_{<\Delta} -\Tilde{\Pi}_{< \Delta} \|\leq \frac{\delta \pi}{4\Delta}.
    \end{equation}
    
\end{restatable}

The proof is postponed to \cref{sec:prooflemma-qula}. Next lemma, also proved in \cref{sec:prooflemma-qula}, quantifies the number of required controlled-$U$ operators to implement the reflection operators. 
\begin{restatable}{lemma}{piprojector}
\label{lemma:grover-unadjusted}
    Let $U$ be the quantum walk operator associated with Unadjusted Langevin algorithm. Under \cref{assumption1,assumption2} the reflection operator $V = e^{i\pi/3}\ket{\pi}\bra{\pi} + (I-\ket{\pi}\bra{\pi}) $ can be implemented with $\epsilon$ accuracy in spectral norm using $\Tilde{O}(\rho^{-1}\beta d L \epsilon^{-1})$ controlled-$U$ operators.
\end{restatable}

Having presented the necessary tools, we next establish the complexity of quantum ULA algorithm.
\qula*
\begin{proof}
We use the same algorithm described in proof of quantum MALA algorithm. That is, we iteratively drive each state $\ket{\mu_i}$ to $\ket{\mu_{i+1}}$ using $\pi/3$ fixed point amplitude amplification algorithm. However, since accessing $U^{\star}$ requires evaluation oracle, we use quantum unadjusted langevin algorithm instead of quantum MALA. The reflection operators can be implemented using $\Tilde{O}(\rho^{-1} \beta d L \epsilon^{-1})$ calls to controlled $U$ operator by \cref{lemma:grover-unadjusted}. Since the length of annealing schedule in \cref{lemma:annealing} is $\Tilde{O}(c_{\mathrm{LSI}}^{-1}\sqrt{d})$, the total complexity is $\Tilde{O}((c_{\mathrm{LSI}}^{-1}\rho^{-1}d^{3/2}\beta\epsilon^{-1})$. Implementing $U$ only requires full gradient oracle constant number of times, we establish the result.

\end{proof}

\subsection{Quantum Unadjusted Langevin Algorithm with Stochastic Gradient}
\label{sec:q-sula}
The construction for the stochastic quantum ULA algorithm is similar to quantum ULA. The stochastic ingredient here is realized by replacing full gradient $\nabla f$ with a stochastic gradient $g_{\ell} = \frac{1}{B}\sum_{k\in S^{\ell}} \nabla f_k$ in implementing the quantum walk operator where $S_\ell$ is a batch randomly uniformly drawn from the set $\{A \subseteq [N]: |A|=B \}$. Here $\ell$ denotes the batch number as the realization of this randomness. We also denote $U_{\ell}$ as the unitary operator constructed with $g_{\ell}$.

We prove that the algorithm constructed with stochastic gradients generates the target state with high probability. The next lemma, proved in \cref{sec:prooflemma-qsula}, quantifies the expectation value of $U_\ell$ over $\ell$ with respect to a deterministic unitary $U$.

\begin{restatable}{lemma}{eulu}
 \label{lemma6}
Let $U_\ell= S(2\sum\limits_x \ketbra*{\psi^{(\ell)}_x}{\psi^{(\ell)}_x}-I)$ be a quantum walk operator where  $\ket{\psi^{(\ell)}_{x}} = \sum_y\sqrt{p^{(\ell)}_{xy}}\ket{y}$ is a quantum state constructed with stochastic gradient $g_\ell$. Let $U = S(2\sum_x \ket{\psi_x}\bra{\psi_x}-I)$. 
Then, we have
\begin{equation}
  \|\mathbb{E}_\ell U_\ell-U\| \leq 6\max_{x\in \Omega}\left\|\mathbb{E}_\ell \ket{\psi_x^{(\ell)}} - \ket{\psi_x} \right\|.  
\end{equation}
\end{restatable}

The next lemma is the application of \cref{lemma6} on quantum Langevin algorithms.
\begin{restatable}{lemma}{eulueta}
\label{lemma7}
Let $U$ be the quantum walk operator for unadjusted Langevin algorithm computed using exact gradients. Let $U_\ell$ be a quantum walk operator for unadjusted Langevin algorithm constructed by computing the gradient on random mini batch $\ell$ of size $B$. Then, under \cref{assumption1,assumption2}, we have 
\begin{equation}
 \|\mathbb{E}_\ell U_\ell-U\|\leq 6\sqrt{2}\eta\beta(LR+G)d^{1/2}/\beta^{1/2},     
\end{equation}
where $G$ is a positive constant that satisfies $\|\nabla f(0)\|\leq G$.
\end{restatable}

The proof is postponed in \cref{sec:prooflemma-qsula}.
The next lemma upper bounds the difference of two random unitary quantum walk operators.
\begin{restatable}{lemma}{ulul}
\label{lemma8}
Let $U_{\ell_1}$ and $U_{\ell_2}$ be two random quantum walk operators constructed with two different stochastic gradients $g_{\ell_1}$ and $g_{\ell_2}$ for unadjusted Langevin algorithm. Then, under \cref{assumption1,assumption2}, we have 
\begin{equation}
     \|U_{\ell_1}-U_{\ell_2}\|\leq 8\sqrt{\eta \beta (LR+G)^2},  
\end{equation}
where $G$ is a positive constant that satisfies $\|\nabla f(0)\|\leq G$.
\end{restatable}
We postpone the proof to \cref{sec:prooflemma-qsula}.

Finally, we present the result and proof for quantum stochastic ULA algorithm in next theorem.
\qsula*
\begin{proof}
Let $U_{\ell}$ be a unitary quantum walk operator defined as,
\begin{align}
    U_\ell = S\left(2\sum_x \ket{\psi^{(\ell)}_x}\bra{\psi^{(\ell)}_x}      -I\right) ,
\end{align}
where $\ket{\psi^{(\ell)}_x}$ is the state,
\begin{align}
    \ket{\psi^{(\ell)}_x} = \sum_y \sqrt{p^{(\ell)}_{xy}}\ket{x}\ket{y},
\end{align}
and $p^{(\ell)}_{xy} =\frac{1}{(4\pi\eta/\beta)}\exp(-\frac{\|y-x+g_{\ell}(x) \|^2}{2\eta/\beta})$, where $g_\ell$ is the stochastic gradient computed on randomly selected data points of size $B$, i.e.,
\begin{align}
    g_\ell(x) \coloneqq \frac{1}{B}\sum_{i\in S_\ell\subseteq [N]} \nabla f_{i}(x).
\end{align}
The number of gradient evaluations for implementing unitary $U_\ell$ is $O(B)$ since we only need to compute gradient on $B$ data points. The key idea in proof of the quantum ULA is the fact that the following operator can be implemented using controlled-$U$ operators:
\begin{align}
    V = e^{i\pi/3}\Pi_{\Delta} + \Pi_{\Delta}^{\perp},
\end{align}
where $\Delta$ is the phase gap of MALA transition density. Suppose that we replace every controlled-$U$ operator with a unitary $U_\ell$. Note that each $U$ in the circuit might be possibly replaced by different unitary due to randomness of stochastic gradients. Let's denote this circuit by $\Tilde{V}$. Now, we show that with high probability $\|V-\tilde{V}\|\leq \epsilon$ for sufficiently small step size. Since the algorithm uses $1/\Delta(U^{\star})$ calls to $U$,
\begin{align}
    \|V - \mathbb{E}(\Tilde{V})\| &\leq \frac{1}{\Delta}\|U^{\star}-\mathbb{E}_{\ell}U_{\ell}\|\\
    &\leq \frac{1}{\Delta}\|U-U^{\star}\|+\|U-\mathbb{E}_{\ell}U_{\ell}\|\\
    &\leq (\rho^{-1}  \sqrt{\beta/\eta}) \eta d L + (\rho^{-1}  \sqrt{\beta/\eta}) (\eta \beta (LR+G)\sqrt{d/B} )\\
    &= \rho^{-1} \eta^{1/2}\beta^{1/2} d L + \rho^{-1} \beta^{3/2}\eta^{1/2} (LR+G)d^{1/2}/B^{1/2}.
\end{align}
Setting $\eta \leq \min(\frac{\epsilon^2\rho^2}{2\beta d^2 L^2},\frac{\epsilon^4\rho^2 B}{4\beta^3 d (LR+G)^2 } )$, we guarantee that,
\begin{align}
    \|V -\mathbb{E}\Tilde{V} \|\leq \epsilon/2.
\end{align}
Next, we use the McDiarmid's inequality to obtain high probability bound:
\begin{align}
    \mathrm{P}(\|\Tilde{V}-\mathbb{E}\Tilde{V}\|\geq \epsilon/2 )&\leq 2\exp(-\frac{\epsilon^2 \Delta}{2\|U_{\ell_1} - U_{\ell_2}\|^2})\\
    &\leq 2\exp(-\frac{\epsilon^2 \rho \eta^{1/2} }{2\beta^{1/2}\|U_{\ell_1} - U_{\ell_2}\|^2})\\
    &\leq 2\exp(-\frac{\epsilon^2 \rho \eta^{1/2} }{128\beta^{1/2}\eta \beta (LR+G)^2})\\
    &= 2\exp(-\frac{\epsilon^2 \rho }{128\beta^{3/2}\eta^{1/2} (LR+G)^2}).
\end{align}
Setting $\eta \leq \frac{\epsilon^4\rho^2}{128^2(LR+G)^4\beta^3}$, guarantees that with at least constant probability, 
\begin{align}
    \|\Tilde{V}-\mathbb{E}\Tilde{V}\|\leq \epsilon/2.
\end{align}
The probability can be boosted in logarithmic number of steps to obtain high probability. Therefore, with high probability, 
\begin{align}
    \|V-\Tilde{V}\|\leq \|V-\mathbb{E}\Tilde{V}\|+\|\mathbb{E}\Tilde{V}-\Tilde{V}\|\leq \epsilon.
\end{align}
Then, to implement the operator $V$ with $\epsilon$ accuracy, we need $\Tilde{O}(1/\Delta)=(\rho^{-1}\sqrt{\eta/\beta}) = \Tilde{O}(\rho^{-2}d\beta/\epsilon^{2})$ calls to $U_{\ell}$. Since each $U_{\ell}$ requires $B=d$ gradient computations and we need to prepare $\Tilde{O}(c_{\mathrm{LSI}}^{-1}\sqrt{d})$ reflections, the total gradient complexity is $\Tilde{O}(c_{\mathrm{LSI}}^{-1}\rho^{-2}d^{5/2}/\epsilon^{2})$.
\end{proof}
\begin{remark}
    In a special scenario where an initial quantum state that has at least constant overlap with $\ket{\pi}$ is provided (e.g. a constant warm state), it is possible to obtain an additional speed up in $d$ dependence by saving up to $O(d^{1/2})$ using a single Markov chain instead of using an annealing schedule.
\end{remark}
\section{Partition Function Estimation}
\label{sec:partitionfunction}
In this section, we describe the method for estimating the partition function for a non-logconcave distribution defined as,
\begin{align}
    Z = \int_{x\in\mathbb{R}^d} e^{-f(x)}\,\dd x.
\end{align}
The partition function can be estimated using the following telescoping product:
\begin{align}
\label{eq:telescope}
    Z = Z_0 \prod_{i=0}^{M-1} \frac{Z_{i+1}}{Z_{i}},
\end{align}
where $Z_i$ is the normalizing constant of the distribution $\mu_i$ defined in \cref{lemma:annealing} and $Z_{M} = Z$. We then approximate $\frac{Z_{i+1}}{Z_i} = \mathbb{E}_{\mu_i}[g_i]$, where
\begin{align}
\label{eq:partition}
    g_i = \exp(\frac{1}{2}\left(\frac{1}{\sigma_i^2}-\frac{1}{\sigma_{i+1}^2}\right)\|x\|^2).
\end{align}

The next lemma is for estimating the normalizing constant $Z_0$.
\begin{lemma}[Lemma 3.1 of \cite{ge2020estimating}]
\label{lemma:partition-initial-estimation}
    Letting $\sigma_1^2 = \frac{\epsilon}{2dL}$, it holds that
    \begin{align}
        \left(1-\frac{\epsilon}{2}\right)\int_{x\in \mathbb{R}^d} \frac{-\frac{1}{2}\|x\|^2}{\sigma_1^2}\,\dd x\leq Z_1\leq \int_{x\in \mathbb{R}^d} \frac{-\frac{1}{2}\|x\|^2}{\sigma_1^2}\,\dd x.
    \end{align}
\end{lemma}

\partitionfunction*
\begin{proof}
    Our algorithm combines the non-destructive mean estimation proposed in \cite{Cornelissen_2023} with our sampling algorithm. Since the partition function can be written in telescoping product given in \cref{eq:telescope}, we need to estimate $\mathbb{E}_{\mu_i}[g_i]$ where the function $g_i$ is given in \cref{eq:partition}. By \cref{lemma:partition-initial-estimation}, we can estimate $Z_0$ with $\epsilon$ accuracy. Then, we show that $g_i$ has constant relative variance for all $i\in [1, M-1]$. First, for $g_{M-1}$,
    \begin{align}
        \frac{\mathbb{E}_{\mu_M}[g_{M-1}^2] }{\mathbb{E}_{\mu_M}[g_{M-1}]^2} &= \mathbb{E}_{\pi}\left[\exp(-\frac{1}{2\sigma_{M-1}^2}\|x\|^2)  \right]\mathbb{E}_{\pi}\left[\exp(\frac{1}{2\sigma_{M-1}^2}\|x\|^2)  \right] \\
        &\leq O(\exp(dL/(m\sigma_{M-1}^4c_{\mathrm{LSI}}^2)))
    \end{align}
    by \cref{lemma10}. Setting $\sigma_M^2 = \Omega(\sqrt{dL/(mc_{\mathrm{LSI}}^2)} )$ implies $ \frac{\mathbb{E}_{\mu_M}[g_{M-1}^2] }{\mathbb{E}_{\mu_M}[g_{M-1}]^2} \leq O(1)$.
    Similarly, for $i\in [1, M-2]$,
    \begin{align}
        \frac{\mathbb{E}_{\mu_i}[g_i^2] }{\mathbb{E}_{\mu_i}[g_i]^2} &=  \frac{\mathbb{E}_{\pi}\left[\exp(-\frac{(1+\alpha)}{2}\|x\|^2 )\right]\mathbb{E}_{\pi}\left[\exp(-\frac{(1-\alpha)}{2}\|x\|^2 )\right]}{(\mathbb{E}_{\pi}[\exp(-\|x\|^2/2)])^2 } \\
        &\leq O(\exp(dL\alpha^2/(mc_{\mathrm{LSI}}^2))
    \end{align}
    by \cref{lemma11}. Therefore, for $\alpha^2 = \Tilde{O}(mc_{\mathrm{LSI}}^2/(dL))$, $ \frac{\mathbb{E}_{\mu_i}[g_i^2] }{\mathbb{E}_{\mu_i}[g_i]^2}\leq O(1)$.
    Having established that the relative variance is constant for all $g_i$, we can derive a modified version of \cite[Lemma 4.4 ]{CCH+19} to combine non-destructive mean estimation and quantum annealing that uses our new sampling algorithms. By \cite[Lemma 4.9 ]{CCH+19}, we can infer that given $\Tilde{O}(\epsilon^{-1})$ copies of $\ket{\Tilde{\mu}_{i-1}}$ such that $\|\ket{\Tilde{\mu}_{i-1}}-\ket{\mu_{i-1}}\|\leq \epsilon$, we can estimate $\Tilde{g}_i$ such that 
    \begin{equation}
        |\Tilde{g}_i - \mathbb{E}_{\mu_i}(g_i) |\leq \epsilon\mathbb{E}_{\mu_i}(g_i)
    \end{equation}
    with success probability $1-o(1)$ using $\Tilde{O}(C(\epsilon))$ steps of quantum walk operator where $C(\epsilon)$ is the complexity of preparing $\ket{\mu}$ with $\epsilon$ error starting from $\ket{\Tilde{\mu}_{i-1}}$. Then using quantum annealing schedule with length $M$, we need to do this process $M$ times with relative error $M/\epsilon$ to estimate $\Tilde{Z}$ such that,
    \begin{align}
        (1-\epsilon)Z\leq \Tilde{Z}\leq (1+\epsilon)
    \end{align}
    with high probability. Therefore, given $\ket{\mu_0}$,  we need to use $MO(C(M/\epsilon))$ steps of quantum walk operator in total. We can choose using quantum mala, quantum ula or stochastic quantum ula algorithms as our quantum walk operators. Since \cref{thm:quantum-mala}, \cref{thm:quantum-ula}, \cref{thm:quantum-sula} are $M$ times complexity of driving $\ket{\mu_i}$ to $\ket{\mu_{i+1}}$ with $\epsilon$ error, we just multiply the results in these theorems by $\Tilde{O}( M/\epsilon)= \Tilde{O}(\sqrt{d}/\epsilon)$ and conclude the proof.
\end{proof}

\section{Acknowledgement}
GO was supported by a seed grant from the Institute of Computational and Data Science (ICDS). XL's research is supported by the National Science Foundation Grants DMS-2111221. CW acknowledges support from National Science Foundation grant CCF-2238766 (CAREER).

\bibliographystyle{plain}
\bibliography{references,annealing}

\newpage
\appendix
\section{Proofs of technical lemmas for Quantum MALA}
\label{sec:prooflemma-qmala}

\phasegap*
\begin{proof}
Let $\gamma(Q) = 1-\lambda_1$ denote the spectral gap of $Q$. Using Cheeger's inequality (i.e. \cite{Cheeger+1971+195+200}), $\gamma(Q)$ can be bounded in terms of the conductance parameter,
\begin{equation}
\label{eq:cheeger}
    \sqrt{2\gamma(Q)} \leq \phi(Q).
\end{equation}
Let $\theta = \arccos|\lambda_1|$. Then we can write,
\begin{align}
\label{eq:phase-gap}
    \Delta(W) \geq |1-e^{2i\theta}|=2\sqrt{1-\lambda_1^2}\geq 2\sqrt{\gamma(Q)}
\end{align}
By combining  \cref{eq:cheeger} and \cref{eq:phase-gap}, we obtain $\Delta(Q)\geq \sqrt{2}\phi(Q)$.
\end{proof}

\section{Proofs of technical lemmas for Quantum ULA}
\label{sec:prooflemma-qula}

\uutilde*
\begin{proof}
We first define the following states,
\begin{align}
    \ket{\psi_x} &= \sum_y \sqrt{p_{xy}}\ket{x}\ket{y}, \\
    \ket{\Tilde{\psi}_x} &= \sum_y \sqrt{\Tilde{p}_{xy}}\ket{x}\ket{y}.
\end{align}
Then, using the definition of quantum walk operators, the spectral norm of difference of operators can be bounded as,
\begin{align}
    \|U - \Tilde{U}\| &= \|S(2\sum_{x\in \Omega}\ket{\psi_x}\bra{\psi_x}-I) - S(2\sum_{x\in \Omega}\ket{\Tilde{\psi}_x}\bra{\Tilde{\psi}_x}-I) \|\\ 
    &\leq 2\left\|\sum_{x\in \Omega}\ket{\psi_x}\bra{\psi_x} - \sum_{x\in \Omega}\ket{\Tilde{\psi}_x}\bra{\Tilde{\psi}_x} \right\|\\
    &\leq 2\left\|\sum_{x\in \Omega}\ket{\psi_x-\Tilde{\psi}_x}\bra{\psi_x}+\sum_{x\in \Omega}\ket{\Tilde{\psi}_x} \bra{\psi_x-\Tilde{\psi}_x}\right\|\\
    &\leq  2\left \|\sum_{x\in \Omega}\ket{\psi_x-\Tilde{\psi}_x}\bra{\psi_x}\right\|+2\left\|\sum_{x\in \Omega} \ket{\psi_x-\Tilde{\psi}_x}\bra{\Tilde{\psi}_x}\right\|,
\end{align}
where the first inequality is due to unitarity of $S$ and the third inequality is due to triangular inequality. Let $\ket{\phi}$ and $\ket{\phi'}$ are the states defined as the maximizers,
\begin{align}
    \left \|\sum_{x\in \Omega}\ket{\psi_x-\Tilde{\psi}_x}\bra{\psi_x}\right\| = \max_{\ket{\phi}}\left \|\sum_{x\in \Omega}\ket{\psi_x-\Tilde{\psi}_x}\bra{\psi_x}\ket{\phi} \right\|,
\end{align}
and 
\begin{align}
    \left\|\sum_{x\in \Omega} \ket{\psi_x-\Tilde{\psi}_x}\bra{\Tilde{\psi}_x}\right\|=\max_{\ket{\phi'}}\left\|\sum_{x\in \Omega} \ket{\psi_x-\Tilde{\psi}_x}\bra{\Tilde{\psi}_x}\ket{\phi'}\right\|.
\end{align}
Notice that, for any $x \in \Omega$, we have  $\bra{\Tilde{\psi}_x}\ket{\Tilde{\psi}_y} = \delta_{xy}$ and $\bra{\psi_x}\ket{\psi_y}=\delta_{xy}$. Therefore, we can write $\ket{\phi}= \sum\limits_{x\in \Omega} c_x\ket{\psi_x} + \ket{\xi}$ and $\ket{\phi'} =\sum\limits_{x\in \Omega} \Tilde{c}_x \ket{\Tilde{\psi}_x} + \ket{\Tilde{\xi}}$ where $\bra{\Tilde{\xi}}\ket{\Tilde{\psi}_x}=\bra{\xi}\ket{\psi_x}=  0$ for all $x\in \Omega$. Hence,
\begin{align}
   \|U - \Tilde{U}\|  &\leq 2\left\|\sum_xc_x\ket{\Tilde{\psi}_x-\psi_x}\right\| +  2\left\|\sum_x\Tilde{c}_x\ket{\Tilde{\psi}_x-\psi_x}\right\|\\
    &\leq 4\max_x\left\|\ket{\Tilde{\psi}_x} - \ket{\psi_x} \right\|.
\end{align}
Finally, we can write,
\begin{align}
    \|\ket{\Tilde{\psi}_x} - \ket{\psi_x}  \| &= \|\sum_{y}(\sqrt{p_{xy}}-\sqrt{\Tilde{p}_{xy}})\ket{x}\ket{y} \|\\
    &= \left(\sum_y (\sqrt{p_{xy}}-\sqrt{\Tilde{p}_{xy}})^2\right)^{1/2}\\
    &\leq \sqrt{2}\|P(x,.)-\Tilde{P}(x,.)\|_{\mathrm{H}}.
\end{align}
where the last step follows from the definition of Hellinger distance.
\end{proof}

\pptilde*
\begin{proof}
For the sake of the proof, we use the lazy version of the Markov chains as it does not change the stationary density. Let $q_{xy} =   \frac{1}{(4\pi\eta/\beta)^{d/2}}\exp(-\frac{\|y-x+\eta \nabla f(x)\|^2}{2\eta/\beta})$, then we can write
\begin{equation}
   p_{xy} = \frac{1}{2} \delta_{xy} + \frac{1}{2}q_{xy},
\end{equation}
and 
\begin{equation}
\begin{aligned}
p^{\star}_{xy} = 
\left\{
    \begin{array}{lr}
        \alpha_x(y)p_{xy} , & \text{if } x \neq y\\
        p_{xy} + \sum\limits_{z\in \Omega} p_{xz}(1-\alpha_x(z))  & \text{if } x =  y
    \end{array}
\right\},
\end{aligned}
\end{equation}
where $\delta_{xy}$ is Kronecker delta function and $\alpha_x(y)$ is the acceptance probability given by
\begin{equation}
\alpha_x(y) =\min \left(1, \frac{\exp(-\beta f(y) -\frac{\|x-y+\eta \nabla f(y)\|^2}{4\eta/\beta})}{\exp(-\beta f(x) -\frac{\|y-x+\eta \nabla f(x)\|^2}{4\eta/\beta})}\right).    
\end{equation}
By this definition, $\alpha_{x}(y)\leq 1$. Suppose $\alpha_{x}(y)\geq 1 - e(x,y)$. Then for $x\neq y$,
\begin{align}
    (\sqrt{p^{\star}_{xy}} - \sqrt{p_{xy}})^2 &= p_{xy}(1-\sqrt{\alpha_{xy}})^2\\
    & \leq p_{xy}(1-\sqrt{1-e(x,y)})^2\\
    &\leq p_{xy} e(x,y)^2,
\end{align}
where the second inequality is due to the fact that for $0\leq x\leq 1$,
\begin{equation}
   1-\sqrt{1-x} = \frac{(1-\sqrt{1-x})(1+\sqrt{1-x})}{(1+\sqrt{1+x})} = \frac{1-(1-x)}{1+\sqrt{1+x}}= \frac{x}{1+\sqrt{1+x}}\leq x.
\end{equation}
For $x = y$,
\begin{align}
    (\sqrt{p^{\star}_{xy}} - \sqrt{p_{xy}})^2 &= p_{xy}\left(\sqrt{1  + \frac{1-\mathbb{E}_{p_{xy}}(\alpha_x(y))}{p_{xy}}} - 1\right)^2\\
    &\leq p_{xy}\left(1+\frac{1-\mathbb{E}_{p_{xy}}(\alpha_x(y))}{2p_{xy}} -1\right)^2\\
     &\leq  \frac{(1-\mathbb{E}_{p_{xy}}(\alpha_x(y)))^2}{4p_{xy}}\\
     &\leq \frac{\mathbb{E}_{p_{xy}}(e(x, y))^2}{2},
\end{align}
where the second inequality follows from $\sqrt{1+x}\leq 1+ \frac{x}{2}$ for $x\geq 0$ and the third inequality holds since $p_{xy}\geq \frac{1}{2}$ for $x=y$ because of laziness of the Markov chains. Therefore, 
\begin{align}
 \int_{y\in \Omega} (\sqrt{p^{\star}_{xy}}-\sqrt{p}_{xy})^2\,\dd y &=\int_{y\in \Omega} \delta_{xy}(\sqrt{p^{\star}_{xy}}-\sqrt{p}_{xy})^2\,\dd y + \int_{y\in \Omega}(1-\delta_{xy})(\sqrt{p^{\star}_{xy}}-\sqrt{p}_{xy})^2\,\dd y\\
 &\leq  \mathbb{E}_{p_{xy}}(e(x,y)^2)+\frac{\mathbb{E}_{p_{xy}}(e(x, y))^2}{2}\\ 
 &\leq  \frac{\mathbb{E}_{q_{xy}}(e(x,y)^2)}{2}+\frac{\mathbb{E}_{q_{xy}}(e(x, y))^2}{8}, 
\end{align}
where the extra factors of $1/2$ and $1/4$ in the second inequality comes from the laziness of the chain. Now, we need to bound $e(x,y)$. Starting from
\begin{align}
    \alpha_{x}(y) &\geq \frac{\exp(-\beta f(y) -\frac{\|x-y+\eta \nabla f(y)\|^2}{4\eta/\beta})}{\exp(-\beta f(x) -\frac{\|y-x+\eta \nabla f(x)\|^2}{4\eta/\beta})}\\
    &=\exp(-\beta(f(y) - f(x)-\frac{2\eta\langle y-x, \nabla f(y)+\nabla f(x) \rangle + \eta^2 \|\nabla f(y)\|^2-\eta^2\|\nabla f(x)\|^2 }{4\eta}  ))\\
    &\geq \exp(-\frac{\beta L \|x-y\|^2}{2} -\frac{\beta \eta^2 \|\nabla f(y)\|^2-\eta^2\|\nabla f(x)\|^2}{4\eta}  ))\\
    & \geq \exp (-\frac{\beta L \|x-y\|^2}{2} -\frac{\beta \eta L(LR+G)\|x-y\| }{2} )\\
    &\geq 1- \frac{\beta L \|x-y\|^2}{2} -\frac{\beta \eta L(LR+G)\|x-y\| }{2}.
\end{align}
The second inequality holds because of the smoothness of $f(x)$ since,
\begin{align}
    f(x)\leq f(y)+\langle y-x, \nabla f(x)\rangle + \frac{L\|x-y\|^2}{2},\\
    f(y)\leq f(x)+\langle x-y, \nabla f(y)\rangle + \frac{L\|x-y\|^2}{2},
\end{align}
which implies the following inequality 
\begin{equation}
    |f(y)-f(x) -\frac{1}{2}\langle y-x, \nabla f(x) + \nabla f(y) \rangle |\leq \frac{L\|x-y\|^2}{2}.
\end{equation}
To obtain the third inequality, we use \cref{lemma:raginsky} to show that,
\begin{align}
    \|\nabla f(x)\| \leq G+ L\|x\|\leq LR+G,
\end{align}
where the last inequality is due to fact that the domain is a ball with radius $R$. Then,
\begin{equation}
    \|\nabla f(x)\|^2 - \|\nabla f(y)\|^2= \|\nabla f(x)-\nabla f(y)\|\|\nabla f(x)+\nabla f(y)\|\leq 2(LR+G)L\|x-y\|.
\end{equation}
Consequently,
$e(x,y)\leq \frac{\beta L \|x-y\|^2}{2} +\frac{\beta \eta L(LR+G)\|x-y\| }{2}$. Finally, we need to bound 
\begin{align}
    \int (\sqrt{p^{\star}_{xy}}-\sqrt{p}_{xy})^2\,\dd y  &\leq \frac{\mathbb{E}_{q(x, .)}(e(x,y)^2)}{2}+\frac{\mathbb{E}_{q(x, .)}(e(x, y))^2}{8}  \\
    &\leq \frac{5}{8}\mathbb{E}_{q_{xy}}(e(x,y)^2)\\
    &\leq \frac{5}{8}\mathbb{E}_{q_{xy}}\left(\frac{\beta L \|x-y\|^2}{2} +\frac{\beta \eta L(LR+G)\|x-y\| }{2}\right)^2\\
    &\leq \frac{5}{8}\beta^2L^2 \mathbb{E}_{q_{xy}}\|x-y\|^4 + \frac{5}{8} \beta^2\eta^2L^2(LR+G)^2\mathbb{E}_{q_{xy}}\|x-y\|^2,
\end{align}
where the second inequality uses Jensen's inequality due to convexity of $e(x,y)$ and the last inequality is due to Young's inequality. Next, we need to compute the expectation values. Notice that since $q_{xy}$ is a Gaussian, the variable $ \frac{\beta \|y-x+\nabla f(x)\|^2}{\eta}$ is a chi-squared distributed random variable with mean $d$ and variance $2d$.
\begin{align}
    \mathbb{E}_{q_{xy}}\|x-y\|^2 &= \mathbb{E}_{q_{xy}}\|x-y+\eta\nabla f(x) - \eta\nabla f(x)\|^2\\
    &\leq 2  \mathbb{E}_{q_{xy}}\|x-y-\eta\nabla f(x)\|^2 + 2\eta^2 \mathbb{E}_{q_{xy}}\| \nabla f(x)\|^2\\
    &\leq 2\eta d/\beta +2\eta^2(LR+G)^2\\
    &\leq 4\eta d/ \beta,
\end{align} 
since the mean of chi squared distribution is  $d$ and $\eta \leq \frac{d}{\beta (LR+G)^2}$. Furthermore,
\begin{align}
    \mathbb{E}_{q_{xy}}\|x-y\|^4 &= \mathrm{Var}_{q_{xy}}\|x-y\|^2 +  (\mathbb{E}_{q_{xy}}\|x-y\|^2)^2\\
    &\leq 2d\eta^2/\beta^2 + (4\eta d/\beta)^2\\
    &\leq 2d\eta^2/\beta^2+16\eta^2d^2/\beta^2
\end{align}
since variance of chi squared distribution is $2d$. Putting things together, we have for any $x\in \Omega$,
\begin{align}
   \|P(x,.)-\Tilde{P}(x,.)\|_H^2 &\leq \frac{5d\eta^2L^2}{4}+10\eta^2d^2L^2 + \frac{5\eta^3d\beta L^2(LR+G)^2}{2}\\
    &\leq 16\eta^2L^2d^2,
\end{align}
for $\eta \leq \frac{d}{\beta (LR+G)^2}$. Hence,  $\|P(x,.)-\Tilde{P}(x,.)\|_H \leq 4 \eta d L$.
\end{proof}

\wwtilde*
\begin{proof}
Let $W = \sum\limits_m e^{2i \phi_m}\ket{\psi_m}\bra{\psi_m}$ where $\phi_0 =0$. Similarly, let $\Tilde{W} = \sum\limits_m e^{2i \Tilde{\phi}_m}\ket{\Tilde{\psi}_m}\bra{\Tilde{\psi}_m}$.
\begin{align}
    \|W\psi_0 - \Tilde{W}\psi_0 \|^2 &= \left\|\sum_m \left(1-e^{2 i \Tilde{\phi}_m}  \right) \ket{\Tilde{\psi}_m}\bra{\Tilde{\psi}_m} \ket{\psi_0}    \right\|^2\\
    &=  \sum_m |1-e^{2i \Tilde{\phi}_m}|^2 \left|\bra{\Tilde{\psi}_m} \ket{\psi_0}\right|^2   \\
    &\geq \sum_{m:\Tilde{\phi}_m\geq \Delta } |1-e^{2i \Tilde{\phi}_m}|^2 \left|\bra{\Tilde{\psi}_m} \ket{\psi_0}\right|^2 \\
    &\geq 16\Delta^2/\pi^2 \sum_{m:\Tilde{\phi}_m\geq\Delta }\left|\bra{\Tilde{\psi}_m} \ket{\psi_0}\right|^2,
\end{align}
where the second inequality is due to $|1-e^{ix}|\geq 2|x|/\pi$ whenever $-\pi\leq x\leq \pi$. Since $ \|W - \Tilde{W}\|\leq \delta$, we have
\begin{equation}
    \sum_{m:\Tilde{\phi}_m<\Delta }\left|\bra{\Tilde{\psi}_m} \ket{\psi_0}\right|^2 \geq 1- \frac{\delta^2\pi^2}{16\Delta^2}.
\end{equation}
Let $\ket{\chi} = \alpha_0 \ket{\psi_0} + \alpha_1 \ket{\psi_0^{\perp}}$ be an arbitray quantum state such that $\alpha_1, \alpha_2\in \mathbb{C}$ and $|\alpha_1|^2+|\alpha_2|^2=1$. Then due to triangular inequality
\begin{align}
        \|\Pi_{< \Delta}\ket{\chi} -\Tilde{\Pi}_{<\Delta }\ket{\chi}  \|\leq |\alpha_0|\|\Pi_{< \Delta}\ket{\psi_0} -\Tilde{|\Pi}_{<\Delta }\ket{\psi_0}\|+|\alpha_1| \|\Pi_{< \Delta}\ket{\psi_0^{\perp}} -\Tilde{\Pi}_{<\Delta }\ket{\psi_0^{\perp}}\|.
\end{align}
We first focus on the first term:
\begin{align}
        \|\Pi_{< \Delta}\ket{\psi_0} -\Tilde{\Pi}_{<\Delta }\ket{\psi_0}  \|&=  \|\ket{\psi_0}-\sum_{m:\Tilde{\phi}_m<\Delta }\ket{\Tilde{\psi}_m}\bra{\Tilde{\psi}_m}\ket{\psi_0} \|\\
        &= \left(2-2\sum_{m:\Tilde{\phi}_m<\Delta }\left|\bra{\Tilde{\psi}_m} \ket{\psi_0}\right|^2\right)^{1/2}\\
        &\leq \frac{\delta\pi}{4\Delta}.
\end{align}

Similarly, for the second term,
\begin{align}
        \||\Pi_{< \Delta}\ket{\psi_0^{\perp}} -\Tilde{\Pi}_{<\Delta }\ket{\psi_0^{\perp}}  \|&=  \left\|\sum_{m:\Tilde{\phi}_m<\Delta }\ket{\Tilde{\psi}_m}\bra{\Tilde{\psi}_m}\ket{\psi_0^{\perp}} \right\|\\
        &= \left(\sum_{m:\Tilde{\phi}_m<\Delta}\left|\bra{\Tilde{\psi}_m} \ket{\psi_0^{\perp}}\right|^2\right)^{1/2}\\
        &= \left(1-\sum_{m:\Tilde{\phi}_m\geq\Delta}|\bra{\Tilde{\psi}_m} \ket{\psi_0}|^2\right)^{1/2}\\
        &\leq \frac{\delta\pi}{4\Delta}.
\end{align}
Since both terms are smaller than $\frac{\delta\pi}{4\Delta}$, we conclude that for any state $\ket{\chi}$, the projectors are at most $\delta\pi/(4\Delta)$ apart in spectral norm.
\end{proof}

\piprojector*
\begin{proof}
Let $P^{\star}$ and $P$ be the transition density of Metropolis Adjusted Langevin algorithm and Unadjusted Langevin algorithm respectively. Let $U^{\star}$ and $U$ be the quantum walk operators built using $P^{\star}$ and $P$ respectively. We can write $U^{\star}$ in spectral form:
\begin{equation}
    U^{\star} = \sum_m e^{2i \phi_m}\ket{\psi_m}\bra{\psi_m}.
\end{equation}
The phase gap $\Delta$ of $U^{\star}$ is defined to be $2|\phi_1|$. Since $P^{\star}$ is a reversible Markov chain, $U^{\star}$ accepts $\ket{\pi}$ as its eigenvector with eigenvalue 1. Furthermore, $\ket{\pi}$ is the unique eigenvector of $U^{\star}$ with eigenvalue 1 (see \cite{Magniez_2011} for more details). Notice that, $R$ can be written as,
\begin{equation}
    V = e^{i \pi/3}\Pi^{\star}_{\Delta} + (I-\Pi^{\star}_\Delta ),
\end{equation}
where $\Pi^{\star}_{\Delta}$ is the projector that projects any quantum state onto the eigenstate of $U^{\star}$ with eigenphase smaller than $\Delta$. This is because the only eigenvector of $U^{\star}$ with phase smaller than $\Delta$ is $\ket{\pi}$. This operator can be implemented in $\epsilon$ accuracy using techniques such as quantum singular value transformation technique introduced in \cite{Gily_n_2019} or phase estimation based method (\cite{Magniez_2011}) using $\Tilde{O}(1/\Delta)$ calls to quantum walk operator. Suppose that we replaced each $U^{\star}$ with $U$ and implement the following operator instead:
\begin{equation}
    \Tilde{V} = e^{i \pi/3}\Pi_{\Delta} + (I-\Pi_\Delta ),
\end{equation}
where $\Pi$ is the projector similarly defined for $U$ which is the quantum walk operator constructed for unadjusted langevin algorithm. Therefore, we can characterize the error,
\begin{align}
    \|V-\Tilde{V}\|&\leq 2\|\Pi^{\star}_{\Delta}-\Pi_{\Delta} \|\\
    &\leq \frac{\|U-U^{\star}\|}{2\Delta/\pi}.
\end{align}
The last inequality follows from \cref{lemma3}. By \cref{lemma:phase-gap} and \cref{lemma13}, $\Delta(U^{\star})\geq c_0\rho\sqrt{2\eta/\beta}$ for step size smaller than $O(d^{-1}\beta^{-1})$. Therefore,
\begin{align}
    \|V -\tilde{V}|& \leq c_0 16\sqrt{2}\pi\eta d L/\Delta\\
        &\leq  c_0 16\pi\sqrt{2\eta\beta}d L/\rho,
\end{align}
where the first inequality is due to $\cref{lemma1}$ and \cref{lemma2}. Therefore, by setting $\eta\leq  \frac{\epsilon^2 \rho^2}{c_016\sqrt{2}\pi d^2L^2\beta}$, we have
\begin{align}
    \|V -\tilde{V}\| \leq \epsilon. 
\end{align}
The total number of calls to $U$ is $\tilde{O}(1/\Delta) = \Tilde{O}(\rho^{-1} \eta^{-1/2}/\beta^{-1/2} )  = \Tilde{O}(\rho^{-1} \beta d L/\epsilon)$.    
\end{proof}

\section{Proofs of technical lemmas for Quantum Stochastic ULA}
\label{sec:prooflemma-qsula}

\eulu*
\begin{proof}
    \begin{align}
        \|\mathbb{E}_\ell U_{\ell}-U\|&\leq  2\|\mathbb{E}_\ell\sum_x \ket{\psi_x^{(\ell)}} \bra{\psi_x^{(\ell)}}-\sum_x\ket{\psi_x}\bra{\psi_x}\|\\
        &= 2\|\mathbb{E}_\ell\sum_x \ket{\psi_x^{(\ell)}} \left(\bra{\psi_x^{(\ell)}}-\bra{\psi_x}\right)+\sum_x\left(\ket{\psi_x^{(\ell)}} -\ket{\psi_x}\right)\bra{\psi_x}\|\\
        &\leq  2\|\mathbb{E}_\ell\sum_x \ket{\psi_x^{(\ell)}} \left(\bra{\psi_x^{(\ell)}}-\bra{\psi_x}\right)\|+2\|\sum_x\left(\ket{\psi_x^{(\ell)}} -\ket{\psi_x}\right)\bra{\psi_x}\|,
\end{align}
where the second inequality follows from triangular inequality. First, we focus on the second term,    
\begin{align}
  \|\mathbb{E}_{\ell}\sum_x\left(\ket{\psi_x^{(\ell)}} -\ket{\psi_x}\right)\bra{\psi_x} \|  = \max_{\ket{\phi}}\|\mathbb{E}_{\ell}\sum_x\left(\ket{\psi_x^{(\ell)}} -\ket{\psi_x}\right)\bra{\psi_x}\ket{\phi}\|.
\end{align}
We can expand the state that maximizes this equation as $\ket{\phi} = \sum_x c_x\ket{\psi_x} +\ket{\xi}$ where $\bra{\xi}\ket{\psi_x} = 0$ for any $x \in \Omega$.
This is true because $\bra{\psi_x}\ket{\psi_y} = \delta_{xy}$. Therefore, $\bra{\psi_x}\ket{\phi} = c_x$.  Then,
\begin{align}
 \|\mathbb{E}_{\ell}\sum_x\left(\ket{\psi_x^{(\ell)}} -\ket{\psi_x}\right)\bra{\psi_x} \|  &= \|\mathbb{E}_{\ell}\sum_x c_x\left(\ket{\psi_x^{(\ell)}} -\ket{\psi_x}\right) \|\\
  & = \left(\sum_x |c_x|^2\|\mathbb{E}_{\ell}\ket{\psi_x^{(\ell)}} -\ket{\psi_x}\|^2   \right)^{1/2}\\
  &\leq \max_x \left(\|\mathbb{E}_{\ell}\ket{\psi_x^{(\ell)}} -\ket{\psi_x}\|^2   \right)^{1/2}\\
  &= \max_x \|\mathbb{E}_{\ell}\ket{\psi_x^{(\ell)}} -\ket{\psi_x}\|.
\end{align}
Again, the first equality is due to $\bra{\psi_x}\ket{\psi_y} = \delta_{xy}$ and the first inequality is due to fact that $\sum_x |c_x|^2\leq 1$. The first term can be written as
\begin{align}
     2\|\mathbb{E}_\ell\sum_x \ket{\psi_x^{(\ell)}} \left(\bra{\psi_x^{(\ell)}}-\bra{\psi_x}\right)\| &= 2\|\mathbb{E}_\ell\sum_x\left(\ket{\psi_x^{(\ell)}}-\ket{\psi_x}\right) \left(\bra{\psi_x^{(\ell)}}-\bra{\psi_x}\right)+ \mathbb{E}_\ell\sum_x \ket{\psi_x} \left(\bra{\psi_x^{(\ell)}}-\bra{\psi_x}\right)\| \\
     &\leq 2\|\mathbb{E}_\ell\sum_x\left(\ket{\psi_x^{(\ell)}}-\ket{\psi_x}\right) \left(\bra{\psi_x^{(\ell)}}-\bra{\psi_x}\right)\| +2\| \mathbb{E}_\ell\sum_x \ket{\psi_x} \left(\bra{\psi_x^{(\ell)}}-\bra{\psi_x}\right)\|.
\end{align}
The second term is bounded by $\max_x \|\psi_x -\mathbb{E}_{\ell}\psi_x ^{(\ell)}\|$ and the first term,
\begin{align}
 \|\mathbb{E}_\ell\sum_x \left(\ket{\psi_x}-\ket{\psi_x^{(\ell)}}\right)\left(\bra{\psi_x}-\bra{\psi_x^{(\ell)}}\right)\| &\leq \max_{x} \|\mathbb{E}_\ell \left(\ket{\psi_x}-\ket{\psi_x^{(\ell)}}\right)\left(\bra{\psi_x}-\bra{\psi_x^{(\ell)}}\right)\|\\
 &= \max_{x}\max_{\ket{\phi}} \|\mathbb{E}_\ell\left(\ket{\psi_x}-\ket{\psi_x^{(\ell)}}\right) \left(\bra{\psi_x}-\bra{\psi_x^{(\ell)}}\right)\ket{\phi}\|\\
 &\leq \max_{x}\|\mathbb{E}_\ell \left(\ket{\psi_x}-\ket{\psi_x^{(\ell)}}\right) \|,
\end{align}
the first inequality is due to the fact that for different $x, y\in \Omega$, $\left(\bra{\psi_x}-\bra{\psi_x^{(\ell)}}\right)\left(\ket{\psi_y}-\ket{\psi_y^{(\ell)}}\right) = 0$ and the last inequality is because $| \left(\bra{\psi_x}-\bra{\psi_x^{(\ell)}}\right)\ket{\phi}|\leq 1$.\\    
\end{proof}

\eulueta*
\begin{proof}
\begin{align}
        \|U - \mathbb{E}_{\ell}U_{\ell}\|^2 &\leq 36 \max_{x\in \Omega}\|\psi_x-\mathbb{E}_{\ell}\psi_x^{(\ell)} \|^2\\
        &\leq 36\max_{x\in \Omega}\|\int_{y\in \mathbb{R}^d} \,\dd y\, (\sqrt{p_{xy}} - \mathbb{E}_{\ell}\sqrt{p_{xy}^{\ell}})\ket{x}\ket{y}  \|^2\\
        &=36\max_{x\in \Omega} \left(\int_{{y\in \mathbb{R}^d}} \,\dd y\, p_{xy} + \int_{y\in \mathbb{R}^d}\,\dd y\,\left(\mathbb{E}_{\ell}\sqrt{p_{xy}^{\ell}}\right)^2-2\mathbb{E}_{\ell}\int_{y\in \mathbb{R}^d} \,\dd y\, \sqrt{p_{xy}p_{xy}^{\ell}}\right)   \\
        &\leq 36\max_{x\in \Omega}\left( \int_{y\in \mathbb{R}^d} \sqrt{p_{xy}} + \mathbb{E}_{\ell}\int_{y\in \mathbb{R}^d} \,\dd y\,p_{xy}^{\ell}-2\mathbb{E}_{\ell}\int_{y\in \mathbb{R}^d} \,\dd y\, \sqrt{p_{xy}p_{xy}^{\ell}}\right)   \\
        &= 36\max_{x\in \Omega}\left(2-2\mathbb{E}_{\ell} \int_{y\in \mathbb{R}^d}\,\dd y\, \sqrt{p_{xy}p_{xy}^{\ell}}\right),
\end{align}
where the first inequality is due to $\cref{lemma6}$ and the second inequality is due to Jensen's inequality since square root is a concave function. 
\begin{align}
      \int\limits_{y\in \mathbb{R}^d} \,\dd y\, \sqrt{p_{xy}p^{\ell}_{xy}} &= \frac{1}{(4\pi\eta/\beta)^{d/2}}\int_{y\in \mathbb{R}^d}\,\dd y\, \exp(-\frac{\|y-x+\eta \nabla f(x)\|^2}{4\eta/\beta} )\exp(-\frac{\|y-x+\eta g_\ell(x)\|^2}{4\eta/\beta} )\\
      &=\frac{1}{(4\pi\eta/\beta)^{d/2}}\int\limits_{y\in \mathbb{R}^d}\,\dd y\, \exp(-\frac{2\|y-x\|^2 +2\eta \langle y-x,\nabla f(x)+g_\ell(x)\rangle +\eta^2 \|\nabla f(x)\|^2+ \eta^2 g_\ell(x)^2 }{4\eta/\beta} )\\
      &=\frac{1}{(4\pi\eta/\beta)^{d/2}}\int\limits_{y\in \mathbb{R}^d} \,\dd y\, \exp(-\frac{\|y-x+\eta (\nabla f(x) +g_{\ell}(x))/2 \|^2 }{2\eta/\beta})\exp(-\frac{\eta^2\|\nabla f(x)-g_{\ell}(x)\|^2}{2\eta/\beta})\\
      &=\exp(-\frac{\eta^2\|\nabla f(x)-g_\ell(x)\|^2}{2\eta/\beta}),
\end{align}
where $\mathbb{E}[g_{\ell}] =\nabla f$. Therefore,
\begin{align}
     \|U - \mathbb{E}U_{\ell}\|^2&\leq 36\max_{x\in \Omega}\left(2-2 \mathbb{E}\exp(-\frac{\eta^2\|\nabla f(x)-g_\ell(x)\|^2}{2\eta/\beta})\right)\\
     &\leq 36(2- 2\exp(-\eta^2\beta^2(LR+G)^2/B))\\
     &\leq 72\eta^2d\beta^2(LR+G)^2/B,
\end{align}
where the first inequality follows from lemma B.2 from \cite{2010.09597},
\begin{align}
\mathbb{E}\exp(\langle a,g_\ell(x) - \nabla f  \rangle)\leq \exp(M^2 \|a\|_2^2/B),
\end{align}
where $M$ is the upper bound on $\|g_\ell(x) - \nabla f(x)\|$ with batch size $B$.
\end{proof}

\ulul*
\begin{proof}
By \cref{lemma1}, the difference of quantum walk operators is bounded by,
\begin{align}
   \|U_{\ell_1}-U_{\ell_2}\|^2 \leq 32 \max_x \|P_{\ell_1} - P_{\ell_2} \|_H^2,
\end{align}
where $P_{\ell_1}$ and $P_{\ell_2}$ are Gaussian transition densities of ULA computed with gradients on mini batches $\ell_1$ and $\ell_2$. This is squared Hellinger distance between two Gaussian distributions with the same variance and different mean. Th
\begin{align}
    \|P_{\ell_1} - P_{\ell_2} \|_H = 1 - \exp(-\frac{\eta^2 \|g_{\ell_1}(x) -g_{\ell_2}(x)\|^2}{2\eta/\beta} )\leq \frac{\eta^2 \|g_{\ell_1}(x) -g_{\ell_2}(x)\|^2}{2\eta/\beta},
\end{align}  
Since $\|\nabla  f(x)\|\leq L\|x\|+G\leq LR+G$, $\|g_{\ell_1}(x) -g_{\ell_2}(x)\|\leq 2(LR+G)$, therefore for any $x\in\Omega$,
\begin{align}
     \|U_{\ell_1}-U_{\ell_2}\|^2\leq 64(LR+G)^2\eta\beta. 
\end{align}
Taking the square root, we obtain the result in the statement.
\end{proof}

\end{document}